\documentclass[10pt,journal]{IEEEtran}

\makeatletter
\def\ifundefined{\@ifundefined}
\makeatother

\newcommand{\chapter}{}

\usepackage{stfloats}
\usepackage{epsfig}
\usepackage{amsmath}
\usepackage{amssymb} 
\usepackage{graphicx}
\usepackage{amsthm}
\usepackage{mathrsfs}
\usepackage{bbold}
\usepackage{multirow}
\usepackage{rotating}
\usepackage{cite}
\usepackage{ifthen}

\newboolean{dcol}
\setboolean{dcol}{true}
\newboolean{suppMat}
\setboolean{suppMat}{false}
\newboolean{longver}
\setboolean{longver}{false}
\newcommand{\bitm}{\begin{itemize}}
\newcommand{\eitm}{\end{itemize}}
\newcommand{\be}{\begin{equation}}
\newcommand{\ee}{\end{equation}}
\newcommand{\bea}{\begin{eqnarray}}
\newcommand{\eea}{\end{eqnarray}}
\newcommand\ba{\renewcommand{\arraystretch}{.7} \left[ \begin{array}{*{6}{@{\hspace{2pt}}r@{\hspace{1.5pt}}}}}
\def\ea{\end{array}\right]}
\def\nn{\nonumber\\}
\newcommand{\bfi}{\begin{figure}}
\newcommand{\efi}{\end{figure}}

\newcommand{\mat}[1]{\begin{bf} #1 \end{bf}}
\newcommand{\fn}{}

\newcommand{\Norm}[2]{|| #1 ||_{{#2}}}
\newcommand{\Ind}[1]{\mathbb{1} \left\{ #1\right\}}
\newcommand{\U}{U}

\newtheorem{thm}{Theorem}  
\newtheorem{lem}{Lemma}        

\newtheorem{pro}{Proposition}

\newenvironment{thmN}[2][Theorem]{\begin{trivlist} 
\item[\hskip \labelsep {\bfseries #1}\hskip \labelsep {\bfseries #2}] \it }{\end{trivlist}}

\newtheorem{defn}{Definition}       

\ifthenelse{\boolean{dcol}}{

}

\begin{document}

\title{On U-Statistics and Compressed Sensing II: Non-Asymptotic Worst-Case Analysis}


\author{
\authorblockN{Fabian Lim${}^*$ and Vladimir Marko Stojanovic}\\
\ifthenelse{\boolean{dcol}}{
\thanks{F. Lim and V. M. Stojanovic are with the Research Laboratory of Electronics, 
Massachusetts Institute of Technology, 77 Massachusetts Avenue, Cambridge, MA 02139.
Email: \{flim,vlada\}@mit.edu.
This work was supported by NSF grant ECCS-1128226.} 
}{
\authorblockA{Research Laboratory of Electronics\\ Massachusetts Institute of Technology, 
77 Massachusetts Avenue, Cambridge, MA 02139\\ 
\{flim,vlada\}@mit.edu 
\thanks{This work was supported by NSF grant ECCS-1128226. } 
}}
}

\maketitle

%

\begin{abstract}
In another related work, U-statistics were used for non-asymptotic ``average-case'' analysis of random compressed sensing matrices. 
In this companion paper the same analytical tool is adopted differently - here we perform non-asymptotic ``worst-case'' analysis.

Simple union bounds are a natural choice for ``worst-case'' analyses, however their tightness is an issue (and questioned in previous works).
Here we focus on a theoretical U-statistical result, which potentially allows us to prove that these union bounds are tight.
To our knowledge, this kind of (powerful) result is completely new in the context of CS.
This general result applies to a wide variety of parameters, and is related to \emph{(Stein-Chen) Poisson approximation}.
In this paper, we consider i) restricted isometries, and ii) mutual coherence.
For the bounded case, we show that $k$-th order restricted isometry constants have tight union bounds, when the measurements $m = \mathcal{O}( k (1 + \log(n/k)))$. 
Here we require the restricted isometries to grow linearly in $k$, however we conjecture that this result can be improved to allow them to be fixed.
Also, we show that mutual coherence (with the standard estimate $\sqrt{(4\log n)/m}$) have very tight union bounds.
For coherence, the normalization complicates general discussion, and we consider only Gaussian and Bernoulli cases here.

\end{abstract}


\ifthenelse{\boolean{dcol}}{
\markboth{Lim and Stojanovic: On U-Statistics and Compressed Sensing II: Non-Asymptotic Worst-Case Analysis}{}
}{
\markboth{IEEE Transactions on Signal Processing, Lim and Stojanovic, Prepared using \LaTeX}{}
}

\section{Introduction}


Recovery analysis in compressed sensing (CS) is usually framed in the context matrix parameters. 
The \emph{restricted isometry constant} is arguably the most commonly studied, as it obtains an important result, that proves that sparse signals with $k$ components can be recovered from the order of $k\log(n/k)$ number of samples from an $n$-dimensional ambient space~\cite{Candes,Donoho}.
On the other hand a wide variety of parameters have been studied, each having its own desirable features (\textit{e.g.}, simplicity, accuracy, etc.). 
To list a few, we have \emph{mutual coherence}~\cite{Grib,Elad,Tropp}, \emph{Karush-Kuhn-Tucker (KKT)} conditions for sparsity pattern recovery (involving matrix \emph{pseudoinverses})~\cite{Fuchs,Tropp,Cand2008,Barga}, and the powerful \emph{null-space property}~\cite{Zhang,Cohen2008,DAspremont2009}.

To handle a wide variety of parameters, we would like a common framework that encompasses their common features.
Here we consider random analysis, and in related work~\cite{Lim1}, we proposed how Hoeffding's \emph{U-statistics} can be good analytical tool.
This is due to a wide availability of general U-statistical theory, and most importantly the fact these statistics model a common feature shared by the CS parameters listed above - that is these parameters are defined combinatorially over all subsets of a fixed size.
Furthermore the tool applies well to the non-asymptotic regime, whereby important practical trends are the focus of recent works with similar random analysis themes~\cite{DTExp, Rudelson2010,Rudel}, or \emph{deterministic}-type CS analysis~\cite{Calder,Gan,Tropp,Barga}. There is however no discussion on U-statistics in CS literature.

In~\cite{Lim1} we show how U-statistics have a natural ``average-case''  interpretation, which we apply to so-called \emph{statistical restricted isometry property} (StRIP) recovery guarantees~\cite{Barga}. 
In this work however, we show how U-statistics also apply well to ``worst-case'' analysis. 
Most CS analysis performed for random matrices are of the ``worst-case'' nature, whereby 
past seminal results established optimal rates and powerful recovery guarantees~\cite{Candes,Donoho,RIPsimp,Blanchard,Bah}.
These analyses mostly involve taking union bounds over a large number of terms, whereby we are interested in an earlier posed question of Blanchard~et.~al.~\cite{Blanchard}: how tight are these union bounds? 

A result that can ascertain tightness of such union bounds could be potentially very useful, since these bounds constitute some of the simplest ways to analyze CS parameters.
While past efforts to answer this question involve innovative bounding methods~\cite{Bah}, and numerical explorations~\cite{Dossal2009}, 
here we discuss a U-statistical result that can answer this question by theoretical proof.
This comes from \emph{(Stein-Chen) Poisson approximation}~\cite{Barbour},~ch. 2, which provides theoretical guarantees on bound tightness, and can be potentially used to show that the union bound cannot be drastically improved.
In other words from a standpoint of characterizing the behavior of CS parameters for recovery analysis, it could be (possibly) shown that simple union bounds are sufficiently good and tight.
To the best of our knowledge, this kind of powerful result has never been investigated before in the CS context.
The U-statistical result is general, and applies to a wide variety of parameters.
For brevity we only discuss two cases here, ``worst-case'' restricted isometries, and mutual coherence - the more complicated null-space property left for future research. 

\textbf{Contributions}:  
We assume throughout that the matrix columns are independently sampled.
We utilize an (non-asymptotic) approximation theorem that predicts how ``worst-case'' U-statistics can be well-approximated by a Poisson distribution (Theorem 2.N).
A good Poisson approximation implies that union bound analyses will essentially be sufficiently tight, where
theoretical approximation error bounds can be given. These error bounds require second-order joint probabilities. 
Denote $n$ and $m$ to be block and measurement sizes, respectively. 
We consider two cases i) restricted isometries and ii) mutual coherence. 
For i) empirical studies suggest good approximation\cite{Dossal2009}.
Here \emph{Ahlsewede-Winter} techniques are used to obtain the necessary joint probabilities.
These techniques lead to simplified arguments, on the other hand there are certain weaknesses that lead to a sub-optimal rate. Nevertheless for bounded matrix entries and $m$ in the order of $k\cdot (1+\log(n/k))$, the approximation error can be shown to exponentially decay in $m$, but at the same time requiring the restricted isometry constants to grow linearly in $k$ (Theorem \ref{cor:DDD}).
We conjecture that this result can be improved. 
For ii), we show that when the mutual coherence is on the order of $\sqrt{(4\log n)/m}$ (the standard estimate, see~\cite{Cand2008}), the Poisson approximation error exponentially decays in $m$ (Theorem \ref{cor:mC}). 
For simplicity for ii) we only consider Gaussian and Bernoulli cases.

\textbf{Organization}:  
We begin with relevant background on CS in Section \ref{sect:params}. In Section \ref{sect:Poi} we present a 
In Section \ref{sect:prob} we derive theoretical bounds related to ``worst-case'' Poission approximation, for both the restricted isometries, and mutual coherence cases. 
We conclude in Section \ref{sect:conc}.



\newcommand{\Real}{\mathbb{R}}

\newcommand{\eigm}{\sigma^2_{\scriptsize \mbox{\upshape min}}}
\newcommand{\eigM}{\sigma^2_{\scriptsize \mbox{\upshape max}}}
\newcommand{\sigm}{\sigma_{\scriptsize \mbox{\upshape min}}}
\newcommand{\sigM}{\sigma_{\scriptsize \mbox{\upshape max}}}
\newcommand{\Tr}{ \mbox{\upshape Tr}}

\newcommand{\Eigm}{\varsigma_{\scriptsize \mbox{\upshape min}}}
\newcommand{\EigM}{\varsigma_{\scriptsize \mbox{\upshape max}}}

\textbf{Notation}: 
The set of real numbers is denoted $\Real$. 
Deterministic quantities are denoted using $a,\mat{a}$, or $\mat{A}$, where bold fonts denote vectors (\textit{i.e.}, $\mat{a}$) or matrices (\textit{i.e.}, $\mat{A}$). 
Random quantities are denoted using upper-case \emph{italics}, where $A$ is a random variable (RV), and $\pmb{A}$ a random vector/matrix.
Let $\Pr\{A\leq a\}$ denote the probability that event $\{A\leq a\}$ occurs.
Sets are denoted using braces, \textit{e.g.}, $\{1,2,\cdots\}$. 
The notation $i,j,\ell,\omega$ is for indexing. 
The notation $\mathbb{E}$ denotes expectation.
We let $\Norm{\cdot}{p}$ denote the $\ell_p$-norm for $p=1$ and $2$. 

\newcommand{\x}{\alpha}
\newcommand{\xbk}{\overline{\pmb{\alpha}}_{k}}

\section{Preliminaries}  \label{sect:params}
\subsection{Compressed Sensing (CS) Theory} \label{ssect:params}

\newcommand{\Sens}{\pmb{\Phi}}
\renewcommand{\S}{\mathcal{S}}
\newcommand{\Fm}{F_{\sigma^2_{\scriptsize \mbox{\upshape min}}}}
\newcommand{\FM}{F_{\sigma^2_{\scriptsize \mbox{\upshape max}}}}
\renewcommand{\a}{a}
\newcommand{\matt}[1]{\pmb{#1}}
\newcommand{\Bas}{\mat{D}}

A vector $\mat{a}$ is said to be $k$-sparse, if at most $k$ vector coefficients are non-zero (\emph{i.e.}, its $\ell_0$-distance satisfies $\Norm{\mat{\a}}{0} \leq k$). 
Let $n$ be a positive integer that denotes {block length}, and let $\matt{\x}=[\x_1,\x_2,\cdots, \x_n]^T$ denote a length-$n$ signal vector with signal coefficients $\x_i$. 
The \emph{best $k$-term approximation} $\xbk $ of $\matt{\x}$, is obtained by finding the $k$-sparse vector $\xbk $ that has minimal approximation error $\Norm{\xbk - \matt{\x}}{2}$.

\newcommand{\col}{\pmb{\phi}}
\newcommand{\y}{b}
\newcommand{\eps}{\epsilon}


Let $\Sens$ denote an $m\times n$ CS sampling matrix, where $m <  n$. 
The length-$m$ \textbf{measurement vector} denoted $\mat{\y}=[\y_1,\y_2,\cdots, \y_m]^T $ of some length-$n$ signal $\matt{\x}$, is formed as $\mat{\y} =\Sens \matt{\x}$.
Recovering $\matt{\x}$ from $\mat{\y}$ is challenging as $\Sens$ possesses a \emph{non-trivial null-space}. 
We typically recover $\matt{\x}$ by solving the (convex) $\ell_1$-\textbf{minimization} problem
\bea
	  \min_{\tilde{\matt{\x}} \in \Real^n} \Norm{\tilde{\matt{\x}}}{1}~~~\mbox{ s. t.   } \Norm{\tilde{\mat{\y}} - \Sens \tilde{\matt{\x}}}{2} \leq \eps. \label{eqn:L1}
\eea
The vector $\tilde{\mat{\y}}$ is a \emph{noisy} version of the original measurements $\mat{\y}$, here $\eps$ bounds the noise error, \emph{i.e.}, $\eps \geq \Norm{\tilde{\mat{\y}}-\mat{\y}}{2}$. 
\ifthenelse{\boolean{dcol}}{ 
Recovery conditions have been considered in many flavors, \textit{e.g.},~\cite{Candes,Donoho,Grib,Elad,Fuchs,Tropp,Cand2008,Barga,Zhang,Cohen2008,DAspremont2009}, mostly by studying parameters of sampling matrix $\Sens$.}{
Recovery conditions have been considered in many flavors,\textit{e.g.},~\cite{Candes,Donoho,Grib,Elad,Fuchs,Tropp,Cand2008,Barga,Zhang,Cohen2008,DAspremont2009}, and mostly rely on studying parameters of the sampling matrix $\Sens$.}


\newcommand{\RIC}{\delta}

For $k \leq n$, the $k$-th \textbf{restricted isometry constant} $\RIC_k$ of an $m\times n$ matrix $\Sens$, equals the smallest constant that satisfies
\bea
	(1-\RIC_k) \Norm{\matt{\x}}{2}^2 \leq \Norm{\Sens\matt{\x}}{2}^2 \leq (1+\RIC_k) \Norm{\matt{\x}}{2}^2, \label{eqn:RIC}
\eea
for any $k$-sparse $\matt{\x} \mbox{ in } \Real^n$.
The following well-known recovery guarantee is stated with respect to $\RIC_k$ in (\ref{eqn:RIC}).


\newcommand{\ConstZero}{c_1}
\newcommand{\ConstOne}{c_2}
\newcommand{\half}{\frac{1}{2}}
\newcommand{\RIPthm}{A}
\begin{thmN}{\RIPthm, \textit{c.f.},~\cite{RIP}} 
Let $\Sens$ be the sensing matrix.
Let $\matt{\x}$ denote the signal vector. 
Let $\mat{\y}$ be the measurements, \textit{i.e.}, $\mat{\y} = \Sens\matt{\x}$. 
Assume that the $(2k)$-th restricted isometry constant $\RIC_{2k}$ of $\Sens$ satisfies $\RIC_{2k} < \sqrt{2} -1$, and further assume that the noisy version $\tilde{\mat{\y}}$ of $\mat{\y}$ satisfies $ \Norm{\tilde{\mat{\y}}-\mat{\y}}{2} \leq \epsilon$. Let $\xbk$ denote the best-$k$ approximation to $\matt{\x}$. Then the $\ell_1$-minimum solution $\matt{\x}^*$ to (\ref{eqn:L1}) satisfies 
\[
	 \Norm{\matt{\x}^* - \matt{\x}}{1} \leq \ConstZero \Norm{\matt{\x} - \xbk}{1} + \ConstOne \epsilon,
\]
for small constants $\ConstZero = 4\sqrt{1+\RIC_{2k}}/(1-\RIC_{2k}(1+\sqrt{2}))$ and $\ConstOne = 2 (\RIC_{2k} (1-\sqrt{2}) - 1)/(\RIC_{2k} (1+\sqrt{2}) - 1) $.
\end{thmN}

\newcommand{\Gram}{\Sens^T\Sens}

\newcommand{\ord}[1]{{(#1)}}
\newcommand{\GramS}{\Sens^T_\S\Sens_\S}
\newcommand{\SensS}{\Sens_\S}
\newcommand{\SensSi}{\Sens_{\S_i}}
\newcommand{\GramSi}{\Sens^T_{\S_i}\Sens_{\S_i}}
\newcommand{\GramSoi}[1]{\Sens_{\S_\ord{#1}}}
\newcommand{\GrampSoi}[1]{\Sens'_{\S_\ord{#1}}}
\newcommand{\N}{N}
\newcommand{\Bin}[2]{{#1 \choose #2}}
\newcommand{\define}{\stackrel{\Delta}{=}}
\renewcommand{\fn}{\footnote{We aim to relax this fairly restrictive assumption in future work.}}



\renewcommand{\fn}{\footnote{For simplicity, we omitted small deviation constants in Theorem B, see~\cite{Candes2004} p. 18 for details.}}

\newcommand{\A}{\pmb{A}}
\newcommand{\func}{\zeta}
\newcommand{\1}{\mathbb{1}}
\newcommand{\LedThm}{B}
Theorem \RIPthm~is very powerful, on condition that we know the constants $\RIC_k$. But because of their combinatoric nature, computing the restricted isometry constants $\RIC_k$ is NP-Hard~\cite{Blanchard}.
The computational difficulty can be seen as follows.
Let $\eigM(\mat{A})$ and $\eigm(\mat{A})$ respectively denote the maximum and minimum, \emph{squared-singular values} of matrix $\mat{A}$.
Denote a function $\func : \Real^{m\times k} \rightarrow \Real$, where for any $\mat{A} \in \Real^{m\times k}$
\bea
\func(\mat{A})= \max(\eigM(\mat{A}) - 1, 1-\eigm(\mat{A})). \label{eqn:g0}
\eea
Let $\S$ denote a size-$k$ subset of indices. 
Let $\SensS$ denote the size $m\times k$ submatrix of $\Sens $, indexed on (column indices) in $\S$.  
We then see from (\ref{eqn:RIC}) that if the columns $\col_i$ of $\Sens$ are properly normalized, \emph{i.e.}, if $\Norm{\col_i}{2}=1$, we deduce that $\RIC_{k}$ satisfies 
\bea 
	\RIC_k &=& \max_\S \func(\SensS),
	\label{eqn:RIC2}
\eea
where the maximization is taken over all $\Bin{n}{k}$ size-$k$ index subsets $\S$. 
For large $n$, the number $\Bin{n}{k}$ is huge. 
To overcome this issue, we may avoid explicitly computing $\RIC_k$ by incorporating \emph{randomization}. 
Let $\A$ denote a random matrix of size $m \times n$. Suppose we sample $\Sens= \A$. 
Let $\A_\S$ denote the size $m\times k$ submatrix of $\A$, indexed on $\S$.
As the mappings $\sigM(\cdot)$ and $\sigm(\cdot)$ corresponding to singular values are $1$-Lipschitz, we have the following well-known measure concentration result\fn.
%

\renewcommand{\fn}{\footnote{More specifically, constant $\eps(k/m,\RIC)$ needs to be set greater than $ \max (\sqrt{1+ \RIC} - (1+ \sqrt{k/m} ),1- \sqrt{k/m} - \sqrt{1-\RIC})$. Note that $\RIC$ must result in the latter quantity being positive.}}
\begin{thmN}{\LedThm, \textit{c.f.},~\cite{Rudel} p. 24, \cite{Candes2004} p. 18} \label{thm:Ledoux}
Assume $\A_\S$ is an $m \times k$ random matrix where $k\leq m$, and assume that the entries $(\A_S)_{ij}$ of $\A_\S$ are both IID with zero mean, \textit{i.e.}, $\mathbb{E}(\A_S)_{ij} = 0$.
Let every $(\A_S)_{ij}$ be either i) Gaussian with variance $1$, or ii) symmetric Bernoulli variables in $\{-1,1\}$.
For every $\eps \geq 0$, we have the probability inequalities
\bea
	\Pr\{\sigm(\A_\S) < \sqrt{m} - \sqrt{k} -  \eps \}  &\leq& \exp(-\eps^2/c_1), \mbox{ and }\nn
	\Pr\{\sigM(\A_\S) > \sqrt{m} + \sqrt{k} +  \eps \}  &\leq& \exp(-\eps^2/c_1), \nonumber
\eea
where $c_1=2$ in the Gaussian case, and $c_1 = 16$ in the Bernoulli case.
\end{thmN}

Let $\A$ have a measure described in Theorem \LedThm, then a union bound over $\Bin{n}{k}$ size-$k$ subsets leads to the following conclusion. 
Let $\1\{\cdot\}$ denote an indicator function.
The probability of sampling $\Sens = (1/\sqrt{m})\cdot \A$, such that the restricted isometry constant of $\Sens$ exceeds $\RIC$, is upper bounded as
\ifthenelse{\boolean{dcol}}{ 
\begin{align}
	\Pr\left\{\sum_{S } \1\{\func(\A_\S) > \RIC m^\half\}  > 0\right\} \leq 2\left(\frac{en}{k}\right)^k e^{\normalsize -\mbox{\normalsize $\frac{m\cdot \eps^2(k/m,\RIC)}{2}$}}, \label{eqn:RICprob}
\end{align}}{
\bea
	\Pr\left\{\sum_{S } \1\{\func(\A_\S) > \RIC m^\half\}  > 0\right\} \leq 2\left(\frac{en}{k}\right)^k 
	\exp\left(\frac{m\cdot \eps^2(k/m,\RIC)}{2} \right), 
	\label{eqn:RICprob}
\eea}
where constant $\eps(k/m,\RIC)$ only depends\fn~on the ratio $k/m$, and constant $\RIC$.
Then if $m c \geq k(1+ \log(n/k))$, for some constant $c$ at most $\eps^2(k/m,\RIC)/2$, the  
RHS of (\ref{eqn:RICprob}) vanishes with increasing $m$.
Then, one claims $\Sens$ has restricted isometry constant at most $\RIC$ with ``large probability''.
\ifthenelse{\boolean{dcol}}{
Note, $\Sens = (1/\sqrt{m})\cdot \A$ does not guarantee the normalization $\Norm{\col_i}{2}=1$ in the Gaussian case, but for simplicity this is usually ignored, see~\cite{Candes,RIPsimp}.}{ 
Note, $\Sens = (1/\sqrt{m})\cdot \A$ does not guarantee the column normalization $\Norm{\col_i}{2}=1$ in the Gaussian case, but for simplicity this is usually ignored, see~\cite{Candes,RIPsimp}.}

Recovery guarantee Theorem \RIPthm~involves ``worst-case'' analysis.
Seen from union bound (\ref{eqn:RICprob}), if \emph{any} one submatrix $\pmb{A}_\S$ satisfies $\func(\pmb{A}_\S) > \RIC m^\half $, the \emph{whole} matrix $\pmb{A}$ is deemed to have restricted isometry constant strictly larger than $\RIC$. 
Still, such union bounds are conceptually simple and thus commonly employed in ``worst-case'' CS analysis. 
They are very useful, whereby in past seminal works they established optimal compression rates and powerful recovery guarantees~\cite{Candes,Donoho,RIPsimp,Blanchard,Bah}.
Hence, it is our interest to present a result that addresses the tightness of union bounds analyses. The question of bound tightness has already been investigated in past works~\cite{Bah,Dossal2009}.
However this work stands apart, by utilizing a mathematical apparatus that is able to theoretically calculate the tightness of such union bounds. 
Such a result would have important theoretical implications, and may potentially void the need for empirical studies and ad-hoc methods.
To our knowledge, such a result has never been discussed before in the context of CS.






\renewcommand{\fn}{\footnote{While~\cite{Lim1} (and references therein) proposed certain benefits of ``average-case'' analysis, we emphasize that ``worst-case'' analyses are useful for important reasons stated in the text. Also, not forgetting that certain applications may strictly require guarantees in the ``worst-case'' sense.}}

The said apparatus is called \emph{U-statistics}, whose concept is introduced in the next subsection.
U-statistical theory is very well-studied, and related work~\cite{Lim1} discusses a different application to ``average-case'' recovery guarantees\fn.

%

\subsection{U-statistics} \label{ssect:UstatStr}

\newcommand{\E}{\mathbb{E}}
\newcommand{\g}{g}
\newcommand{\ZO}{\{0,1\}}

U-statistics were invented in the late 40's by Hoeffding as a theory for non-parametric testing~\cite{Hoeffding1948}.
A function $\func : \Real^{m \times k} \rightarrow \Real$ is said to be a \textbf{kernel}, 
if for any $\mat{A},\mat{A}' \in \Real^{m \times k}$, we have $\func(\mat{A}) = \func(\mat{A}')$ if matrix ${\mat{A}}'$ can be obtained from $\mat{A}$ by \emph{column reordering}.
U-statistics are associated with functions $g : \Real^{m \times k} \times \Real \rightarrow \ZO$ known as \textbf{indicator kernels}.
In this paper we only consider indicator kernels $g$ of the form $g(\mat{A},a) = \Ind{\func(\mat{A}) > a}$.
Examples of indicators can be constructed with $\func $ equals (\ref{eqn:g0}), as well as $\func=\eigM$ and $\eigm$.
%
The following definition slight differs from that of~\cite{Lim1}.

\newcommand{\V}{U}
\begin{defn}[Indicator Kernel U-Statistics] \label{def:Ustat}
Let $\A$ be a random matrix with $n$ columns. Let $\Sens$ be sampled as $\Sens=\A$. 
Let $g : \Real^{m \times k} \times \Real \mapsto \ZO$ be a indicator kernel. For any $a\in \Real$, the following quantity
\bea
	 \U_n(a) \define \frac{1}{\Bin{n}{k}} \sum_{\S} g(\Sens_{\S}, a) \label{eqn:Ustat}
\eea
is a U-statistic of the sampled realization $\Sens=\A$, corresponding to the kernel $g$.
In 
(\ref{eqn:Ustat}), the matrix $\SensS$ is the submatrix of $\Sens$ indexed on column indices in $\S$, and the sum takes place over all subsets $\S$ in $\{1,2,\cdots, n\}$. 
Note, $0 \leq U_n(a) \leq 1$.
\end{defn}

\newcommand{\BargThm}{C}
\newcommand{\Reg}{\theta_n}
\newcommand{\noisesd}{c_Z}
\newcommand{\CandThm}{D}


\newcommand{\gauss}{\int_{-\infty}^a (1/\sqrt{2\pi}) e^{-\t^2/2}d\t}
\renewcommand{\t}{t}

\section{Poisson approximation theorem: ``worst-case'' behavior} \label{sect:Poi}

This section states the U-statistic result, that allows us to compute the tightness of union bounds used in ``worst-case'' analysis. 
For illustration here we use restricted isometries. 

One way of putting resticted isometries (\ref{eqn:RIC2}) in the U-statistical context, is to set $\func$ as in (\ref{eqn:RIC}) and consider
the event $\{\U_n(a) =0 \}$ that occurs when $\func(\Sens_\S)$ \emph{never} exceeds $a$, \textit{i.e.},  
\bea
\Pr\{\U_n(a) =0 \} = \Pr\{\max_\S \func(\A_\S) \leq a \}. \label{eqn:Umax}
\eea
However, here we follow~\cite{Lim1} and consider both maximum and minimum squared eigenvalues in $\func$ separately.
For $\func = \sigM^2$ and $\func = \sigm^2$ respectively, we consider the complementary events (of the LHS of (\ref{eqn:Umax}))
\bea
\Pr\{\U_n(a)  > 0 \} &=& \Pr\{\max_\S \sigM^2(\A_\S) > a \}, \nn
\Pr\{\U_n(a)  > 0 \} &=& \Pr\{\max_\S -\sigm^2(\A_\S) >-a \}, \label{eqn:Umax2}
\eea
For $a= \RIC_k$, these two events respectively correspond to violation of the upper, and lower, inequalities of (\ref{eqn:RIC}).
Observe how these bounds are similar to the previous union bound (\ref{eqn:RICprob}).

\newcommand{\mtail}{p}
\newcommand{\jtail}{q}
\newcommand{\z}{a}
\newcommand{\Lam}{\lambda_n}
\newcommand{\kernel}{\zeta}
Techniques developed for estimating $\Pr\{\U_n(a) =0 \}$, see (\ref{eqn:Umax}), fall under the umbrella term \emph{Poission approximation}, see~\cite{Barbour,Trust2011,Lao}. The terminology comes from similarities with the Poisson limit of a binomial distribution, see~\cite{Barbour}, ch. 1. To illustrate the last point, consider the special case $k=1$, and let $\col_i$ denote the $i$-th column of $\Sens$. Then $\U_n(a)$ equals the average $\U_n(a) = \frac{1}{n} \sum_{i=1}^n g(\col_i,a) $, where we consider subsets $\S$ of size-$1$ of the form $ \S= \{i\}$, \textit{i.e.}, $\Sens_\S = \col_i$. 
Suppose we sample $\Sens=\A$ where random matrix $\A$ has $n$ IID columns $\A_i$.
Then for any $\a$, we have $g(\A_i,a)$ to be Bernoulli distributed with probability  $p(a)$, recall $p(a) = \E g(\A_i,a) = \Pr\{\func(\A_i) > a \}$. 
Furthermore $U_n(a)$ has the \emph{binomial distribution}, which is well approximated by the expression $(n p(a))^j e^{-n p(a)} /(j!)$ for small probability $p(a)$ and index $j$, see~\cite{Barbour}, ch. 1. The previous expression is nothing but the \emph{Poisson distribution} function with parameter $n p(a)$.

For $k=1$, Poisson approximation naturally holds for U-statistics, but the extension to general $k \geq 1$ is non-trivial.
Let both $\S$ and $\mathcal{R}$ denote index subsets of size $k$. 
Define
\bea
\jtail_i(\z)  &=&  \Pr\{\kernel(\A_\S) > \z,\kernel(\A_{\mathcal{R}}) > \z \} \label{eqn:tail}
\eea 
for $1 \leq i < k$ and subsets $\S,\mathcal{R}$ where $|\S \cap \mathcal{R}| = i$.
Let $\Lam(\z)$ denote the sum of all tail probabilities, \textit{i.e.}, let $\Lam(\z) = \Bin{n}{k} \mtail(\z) $.

\renewcommand{\define}{=}
\newcommand{\PoiThm}{2.N}
\begin{thmN}{\PoiThm, \textit{c.f.},~\cite{Barbour}, see p. 35}
Let $\A$ be an $m\times n$ random matrix, whereby the columns $\A_i$ are IID. 
Let $\func$ be a kernel that maps $\Real^{m\times k} \rightarrow \Real$. 
Let $g$ be an indicator kernel that maps $\Real^{m\times k} \times \Real \rightarrow \{0,1\}$ that satisfies $g(\mat{A},a) = \Ind{\func(\mat{A}) > a}$, and let $p(a)=\E g(\A_\S,a) = \E \U_n(a)$. Let $\U_n(a)$ be a U-statistic of sampled realization $\Sens=\A$ corresponding to indicator kernel $g$.
For all $1\leq i < k$, let $q_i(a)$ be defined as in (\ref{eqn:tail}). Let $\Lam(\z)\define \Bin{n}{k} p(a)$. For some $a\in \Real$ whereby $\mtail(\z) > 0$, the probability $\Pr\{\U_n(a) = 0\} = \Pr\{\max_\S \func(\A_\S) \leq a \}$ is approximated by the function $\exp(-\Lam(\z))$ as follows
\[
	  \left|\Pr\{\max_\S \func(\A_\S)\leq \z\} - \exp(-\Lam(\z)) \right| \leq \eps_n(\z),
\]
where the approximation error $\eps_n(\z)$ is given as 
\ifthenelse{\boolean{dcol}}{ 
\bea
	\eps_n(\z) &=& (1-e^{-\Lam(\z)})\left\{\mtail(\z) \left[\Bin{n}{k}-\Bin{n-k}{k} \right] \right. \nn 
	&& \left.+ \sum_{r=1}^{k-1}\Bin{k}{r}\Bin{n-k}{k-r} \mtail(\z)^{-1} \cdot \jtail_r(\z) \right\}. \label{eqn:errPA}
\eea
}{
\bea 
	\eps_n(\z) = (1-e^{-\Lam(\z)})\left\{\mtail(\z) \left[\Bin{n}{k}-\Bin{n-k}{k} \right] + \sum_{r=1}^{k-1}\Bin{k}{r}\Bin{n-k}{k-r} \mtail(\z)^{-1} \cdot \jtail_r(\z) \right\}. \label{eqn:errPA}
\eea}
\end{thmN}
In the sequel, Theorem \PoiThm~will lead to calculating the tightness of union bounds.
The proof uses Stein-Chen techniques and is rather lengthly, thus we refer the reader to~\cite{Barbour}, ch. 2.
Similar to Theorem 1 presented in~\cite{Lim1}, Theorem \PoiThm~also requires IID columns.
The quantities (\ref{eqn:Umax2}) of interest is approximated by the function $1-\exp(-\Lam(\z))$ up to error $\eps_n(\z)$ in (\ref{eqn:errPA}).
Note that Theorem \PoiThm~is a non-asymptotic result because of explicit dependence on system sizes $k,m,n$.

\ifthenelse{\boolean{dcol}}{ 
\begin{figure}[!t]
	\centering
	  \epsfig{file=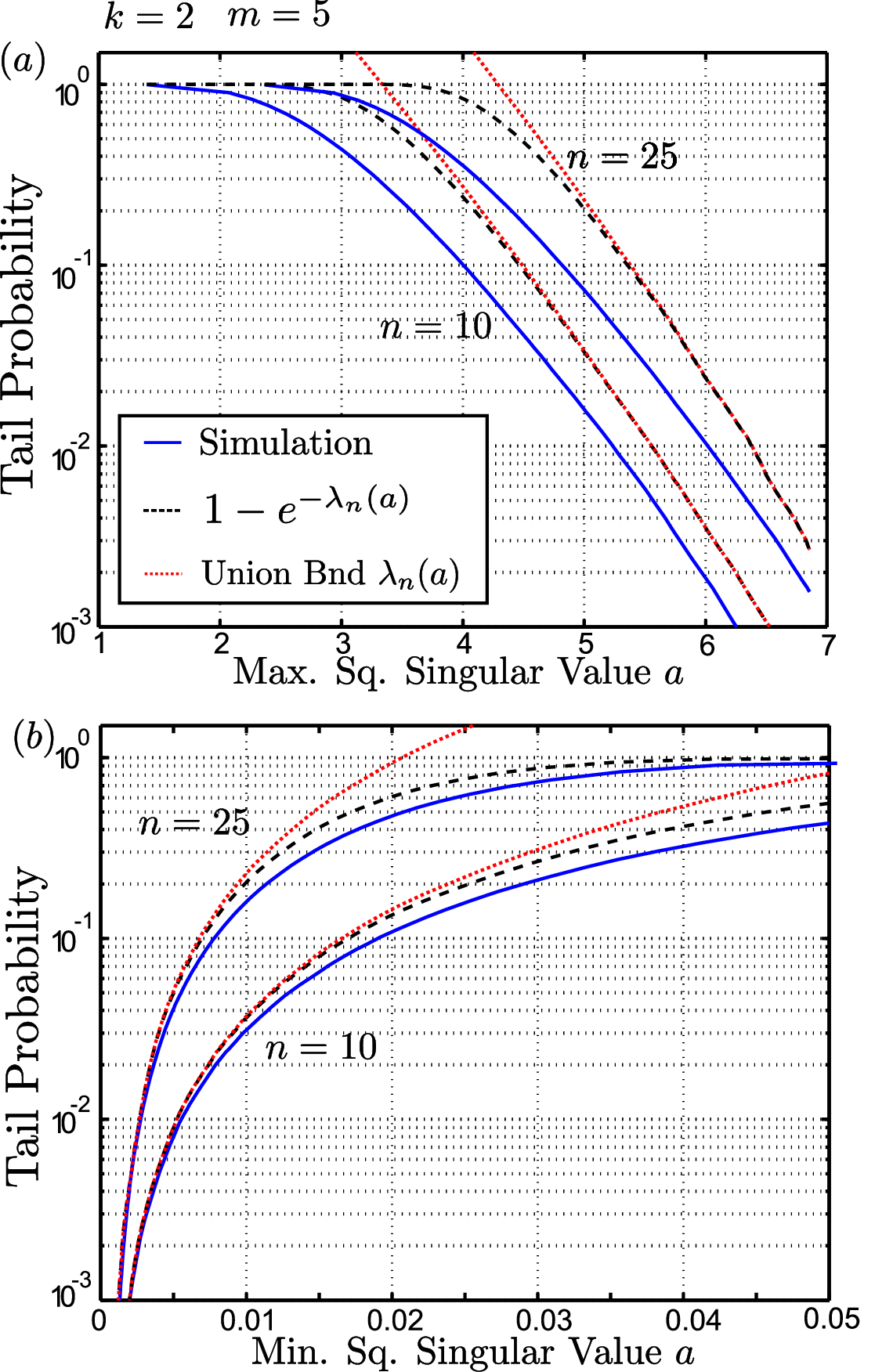,width=.7\linewidth}
		\caption{Gaussian measure. Empirical tail probability $\Pr\{\max_\S \func(\A_\S)> \z\}$ is shown, where $\func=\eigM$ in $(a)$, and $\func=-\eigm$ in $(b)$, respectively corresponding to the the maximum and minimum squared singular values.} 
		\label{fig:GaussRI_extreme}
		\vspace{-15pt}
\end{figure}}
{
\begin{figure}[!t]
	\centering
	  \epsfig{file=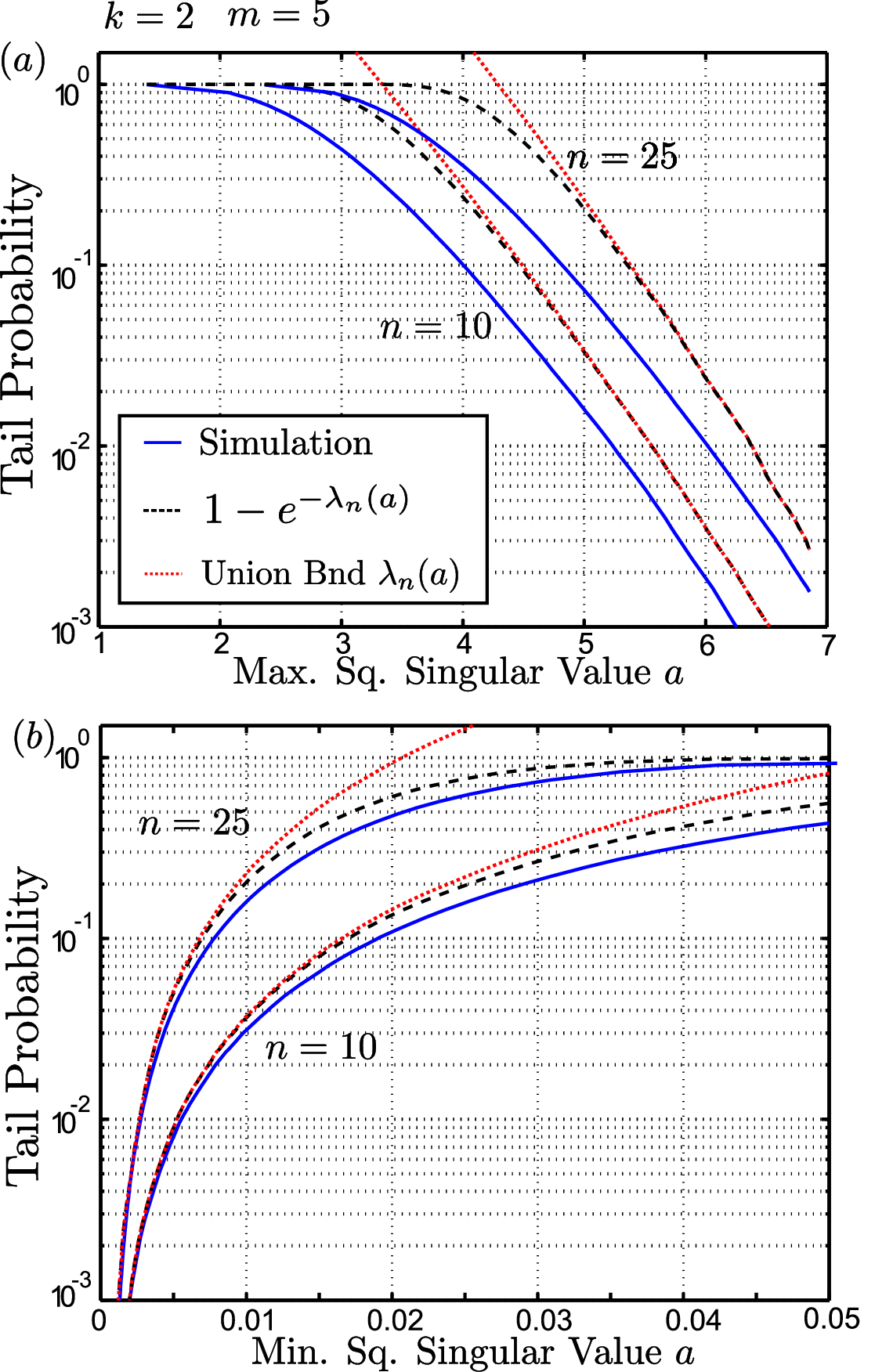},width=.9\linewidth}
		\caption{Gaussian measure. The tail probability $\Pr\{\max_\S \func(\A_\S)> \z\}$ is shown, where $\func=\eigM$ in $(a)$, and $\func=-\eigm$ in $(b)$, respectively corresponding to the the maximum and minimum squared singular values.} 
		\label{fig:GaussRI_extreme}
\end{figure}}
First, we illustrate Theorem \PoiThm~using some simulation results. We draw size $5 \times n$ random matrices $\A$, where the columns $\A_i$ are drawn IID. For $\func = \eigM$, Figure \ref{fig:GaussRI_extreme}$(a)$ shows the tail distribution $\Pr\{\max_\S \eigM(\A_\S)>  \z\}$. 
This is obtained by empirical simulation, performed for $k=2$ and two block lengths $n=10$ and $n=25$. 
Figure \ref{fig:GaussRI_extreme}$(a)$ reveals reasonably good approximation for all shown values for $a$ (compared to the function $1-e^{\Lam(a)}$), within a factor of $2$-$4$. 
The approximation is observed to improve in the higher part of the tails (\textit{i.e.} for larger values of $a$). 
For $\func = -\eigm$, Figure \ref{fig:GaussRI_extreme}$(b)$ presents similar empirical comparisons for $\Pr\{\min_\S \eigm(\A_\S) <   \z\}$ in (\ref{eqn:Umax2}). In this case we notice better approximation, the differences become hardly noticeable. 
The extremely small $k,m, n$ values chosen in this experiment suggest Theorem \PoiThm~works well for non-asymptotics. 


The quantity $\Lam(\z)$ defined above,
is in fact a tail probability union bound.
Notice that $\Pr\{\U_n(a) > 0\}\leq \Lam(\z)$, see (\ref{eqn:RICprob}), and we display $\Lam(\z)$ in Figure \ref{fig:GaussRI_extreme}. 
We claim that Theorem \PoiThm~can in fact be used to evaluate of the tightness of the union bound, and shows us (if at all) how much the bound can be improved. 
This is because even though the error (\ref{eqn:errPA}) is given w.r.t. $1-e^{-\Lam(\z)}$ and not $\Lam(\z)$, note that these functions are close for the region of interest, \textit{i.e.}, $1-e^{-\Lam(\z)}= \Lam(\z) + o(\Lam(\z))$ for $\Lam(\z) \rightarrow 0$. 

Union bound analyses only make sense when $\Lam(\z)$ is small.
The first term $[\Bin{n}{k} - \Bin{n-k}{k}] \cdot \mtail(\z)$ from (\ref{eqn:errPA}) must also be small, since is at most $\Lam(\z)$.
We are mostly concerned with the second term, which depends on the joint tail probabilities $q_i(\z)$.
Using the fact $q_{k-1}(\z) \geq q_i(\z)$ for $i \leq k-1$, see~\cite{Trust2011} Lemma 1, we further derive another useful form of (\ref{eqn:errPA})
\ifthenelse{\boolean{dcol}}{ 
\begin{align}
	\eps_n(a) \leq& ~(1-e^{-\Lam(\z)})\left\{\mtail(\z) \left[\Bin{n}{k}-\Bin{n-k}{k} \right]  \right\}  \nn
	 &+ 2^k \cdot \left( \frac{e(2k-1)}{k-1} \right)^{k-1} \cdot \left( \frac{e n}{2k-1} \right)^{2k-1} q_{k-1}(a),
	\label{eqn:errPA2}
\end{align}}{
\bea
	\eps_n(a) &\leq& (1-e^{-\Lam(\z)})\left\{\mtail(\z) \left[\Bin{n}{k}-\Bin{n-k}{k} \right]  \right\}  \nn
	 &&+ 2^k \cdot \left( \frac{e(2k-1)}{k-1} \right)^{k-1} \cdot \left( \frac{e n}{2k-1} \right)^{2k-1} q_{k-1}(a),
	\label{eqn:errPA2}
\eea}
which only depends on a single joint term $q_{k-1}(a)$. 
The exact details of the derivation is given in Appendix \ref{sup:error}.
The constant term in front of $q_{k-1}(a)$ has exponent at most $(2k-1) \cdot [ 1 + \log (n/(k-1))]$.
This easily follows by the bound $(2k-1)/(k-1) = 2 + 1/(k-1) > 2$, and simple arithmetic. 
Hence as the exponent of $\Bin{n}{k}$ (in front of $p(a)$) in $\Lam(a)$ is at most $k (1 + \log(n/k))$, the error (\ref{eqn:errPA2}) will be small if we can ensure that the joint probability $q_{k-1}(a)$ drops ``twice as fast'' as $p(a)$.


%


In the next section, we evaluate the (non-asymptotic) approximation error $\eps_n(a)$ given by Theorem \PoiThm,~for two important CS parameters taken from well-studied ``worst-case'' analyses: restricted isometries (see Subsection \ref{ssect:RI}) and mutual coherence (see Subsection \ref{ssect:MC}). 
The main theorems/conclusions will be given for both cases.
We leave the more complicated null-space property for future work.

\newcommand{\Iv}{^{-1}}
\newcommand{\R}{\mathcal{R}}

\newcommand{\s}{\pmb{\beta}}
\newcommand{\jp}{\ell}
\newcommand{\Rj}{{\R\setminus \{j\}}}

\section{Poisson Approximation Error \& Non-asymptotic Probability Estimates} \label{sect:prob}



\subsection{Restricted isometries case} \label{ssect:RI}

\newcommand{\row}[1]{{[#1]}}
\renewcommand{\eigm}{\varsigma_{\scriptsize \mbox{\upshape  min}}}
\renewcommand{\eigM}{\varsigma_{\scriptsize \mbox{\upshape max}}}
\renewcommand{\Tr}{ \mbox{\upshape Tr}}
We focus on the case of restricted isometries, \textit{i.e.} when $\func$ is set to equal $\sigM^2$ and $-\sigm^2$ respectively, as in (\ref{eqn:Umax2}). 
Estimates for the joint tail probabilities $q_i(a)$, are not as well-addressed as the marginals $\Pr\{\kernel(\A_\S) > \z\}$ (denoted $p(a)$). 
The following proposition presents such estimates. 
Also for any two Bernoulli distributions with probabilities $a$ and $b$, let $\mathcal{D}(a||b)$ denote the \textbf{binary information divergence}, i.e. $\mathcal{D}(a||b) = a \log(a/b) + (1-a) \log ((1-a)/(1-b))$. 
Note that $\log$ here indicates \emph{natural log}.
For any matrix $\mat{A}$ with entries $a_{ij}$, the \textbf{$i$-th row outer product} of $\mat{A}$ equals the $k\times k$ matrix with entries $a_{i \ell} \cdot a_{i\omega}$. 

\newcommand{\taup}{\tau_p}
\newcommand{\tauAB}{\tau_q}
\newcommand{\taupM}{\tau_{p,\scriptsize \mbox{\upshape max}}}
\newcommand{\taupm}{\tau_{p,\scriptsize \mbox{\upshape min}}}
\newcommand{\pO}{c_1}
\newcommand{\pT}{c_2}

\newcommand{\Am}{\mat{C}}
\newcommand{\Bm}{\mat{D}}

\ifthenelse{\boolean{dcol}}{ 
\begin{figure*}[!b]
	\newcounter{tempequationcounter}
	\normalsize
	\setcounter{tempequationcounter}{\value{equation}}
	\hrulefill
	\begin{align}
	\setcounter{equation}{12}
		&\Pr\{\sigM^2(\Sens_\S) > \z,\sigM^2(\Sens_{\mathcal{R}}) > \z \}  \leq 
		k^2  \exp\left(-m\cdot 2\left(\mathcal{D}\left(\left.\left.\frac{a}{k}\right|\right|\frac{\tauAB}{\taupM}\right) + c_3\left(a,k,\frac{\tauAB}{\taupM},\frac{\taupM^2}{\tauAB}\right)\right)\right)\nn
		&\Pr\{\sigm^2(\Sens_\S) < \z,\sigm^2(\Sens_{\mathcal{R}}) < \z \} \leq 
		k^2  \exp\left(-m\cdot 2\left(\mathcal{D}\left(\left.\left.\frac{a}{k}\right|\right|\frac{\tauAB}{\taupm}\right) + c_3\left(a,k,\frac{\tauAB}{\taupm},\frac{\taupm^2}{\tauAB}\right)\right)\right)
		\label{eqn:joint}
	\end{align}
	\setcounter{equation}{\value{tempequationcounter}}
\end{figure*}}

\begin{pro} \label{pro:joint}
Let $\pmb{A}$ be an $m \times n$ random matrix, whereby the columns $\A_i$ of $\A$ are identically distributed. Let every entry $A_{ij}$ of $\A$, satisfy the bound $|A_{ij}| \leq 1/\sqrt{m}$, such that the columns $\A_i$ are normalized as $\Norm{\A_i}{2}\leq 1$. Let the rows $[A_{i1},A_{i2},\cdots, A_{in}]$ of $\A$ be IID.

Let $\S,\R$ be size-$k$ index subsets, whereby $\S$ and $\R$ intersect in exactly $i$ positions, i.e., $|\S \cap \R| = i$. Let $\pmb{X}$ and $\pmb{Y}$ equal the first row outer products of the matrices $\sqrt{\frac{m}{k}} \A_\S$ and $\sqrt{\frac{m}{k}} \A_\mathcal{R}$, respectively.
Let $\tauAB$ denote a constant that satisfies 
\bea
	\tauAB \cdot \Tr(\Am) \Tr(\Bm) \geq \E \Tr(\Am\pmb{X}) \Tr(\Bm\pmb{Y}), \label{eqn:tauAB}
\eea
where $\Am,\Bm$ can be any positive semidefinite matrices of size $k\times k$, and $\Tr(\cdot)$ denotes trace. Also define constants $\taupM$ and $\taupm$ as follows $\taupM \define \eigM(\E\pmb{X})$ and $\taupm \define \eigm(\E\pmb{X})$, where
$\eigM$ and $\eigm$ denote maximum and minimum eigenvalues, respectively.
%

Assume $\max(\taupM,\taupm) \leq \sqrt{\tauAB} $.
Let $\mathcal{D}(\cdot || \cdot)$ denote binary information divergence.
For all $1\leq i < k$ such that $|\S \cap \R| = i$, 
\ifthenelse{\boolean{dcol}}{
the joint tail probability bounds (\ref{eqn:joint}) (see page bottom) hold
}{ 
the following joint tail probability bounds hold
\bea
\!\!\!\!\!\! \Pr\{\sigM^2(\Sens_\S) > \z,\sigM^2(\Sens_{\mathcal{R}}) > \z \}
\!\!\! &\leq& \!\!\!
k^2  \exp\left(-m\cdot 2\left(\mathcal{D}\left(\left.\left.\frac{a}{k}\right|\right|\frac{\tauAB}{\taupM}\right) + c_3\left(a,k,\frac{\tauAB}{\taupM},\frac{\taupM^2}{\tauAB}\right)\right)\right),\nn
\!\!\!\!\!\! \Pr\{\sigm^2(\Sens_\S) < \z,\sigm^2(\Sens_{\mathcal{R}}) < \z \}
\!\!\! &\leq& \!\!\!
k^2  \exp\left(-m\cdot 2\left(\mathcal{D}\left(\left.\left.\frac{a}{k}\right|\right|\frac{\tauAB}{\taupm}\right) + c_3\left(a,k,\frac{\tauAB}{\taupm},\frac{\taupm^2}{\tauAB}\right)\right)\right),\nn
\label{eqn:joint}
\eea}
for $k\cdot \taupM < a < k$ and $0 < a< k \cdot \taupm $ respectively, and
where the constant $c_3 =  c_3(a,k,\pO,\pT)$ satisfies
\ifthenelse{\boolean{dcol}}{ 
\begin{align}
	\setcounter{equation}{13}
	c_3(a,k,\pO,\pT) &= \frac{a}{k}\log \left(\frac{c_4 + \frac{a}{k}}{\frac{a}{k}} \right) \nn
	&-\frac{1}{2}\log\left( \pT(1+c_4)^2 + \left(1-\pT\right) \left(\frac{1-\frac{a}{k}}{1-\pO}\right)^2 \right)  \label{eqn:jointc3}
\end{align}}{
\bea
	c_3(a,k,\pO,\pT) &=& \frac{a}{k}\log \left(\frac{c_4 + \frac{a}{k}}{\frac{a}{k}} \right)
	-\frac{1}{2}\log\left( \pT(1+c_4)^2 + \left(1-\pT\right) \left(\frac{1-\frac{a}{k}}{1-\pO}\right)^2 \right)  \label{eqn:jointc3}
\eea}
\ifthenelse{\boolean{dcol}}{ 
and the constant $c_4 = c_4(a,k,\pO,\pT)$ satisfies
\bea
 \!\!\!c_4(a,k,\pO,\pT) \!\!\!&=& \!\!\!\frac{1}{2}\sqrt{1 + \frac{4(\pT^{-1}-1)(1-\frac{a}{k})\frac{a}{k}}{(1-\pO)^2}}-\frac{1}{2}. \label{eqn:jointc4}
\eea}{ 
and the constant $c_4 $ satisfies
\bea
c_4 = c_4(a,k,\pO,\pT) &=& \frac{1}{2}\sqrt{1 + \frac{4(\pT^{-1}-1)(1-\frac{a}{k})\frac{a}{k}}{(1-\pO)^2}}-\frac{1}{2}. \label{eqn:jointc4}
\eea}
\end{pro}

Proposition \ref{pro:joint} provides upper bounds on $q_i(a) $ or $ \Pr\{\kernel(\A_\S) > \z,\kernel(\A_{\mathcal{R}}) > \z \}$, for both cases $\func = \sigM^2$ and $\func = -\sigm^2$. 
This result requires an estimate for the constant $\tauAB$ in (\ref{eqn:tauAB}).
While Proposition \ref{pro:joint} does not assume independent columns $\pmb{A}_i$, however Theorem \PoiThm~does. Under this additional column independence assumption, we claim that we can take
\bea
\tauAB = \frac{\beta}{k^2}~~~~ \mbox{ and } ~~~~\taupM = \taupm = \frac{\beta'}{k}. \label{eqn:ber_claim}
\eea
for some constants $\beta,\beta'$. 
The latter claim is easily verified to be true, whereby in this case $\E \pmb{X}$ (see Proposition \ref{pro:joint}) equals an identity matrix scaled by some $\beta'/k$, specifically $\beta' = m \cdot (\E A_{ij})^2$.
If $A_{ij}$ are Bernoulli $\{-1/\sqrt{m},1/\sqrt{m}\}$ then $\beta'=1$ and we claim $\beta = 3$.
The former will be clarified in the upcoming Proposition \ref{pro:tau_est} in this subsection. 
Note, to meet the $\max(\taupM, \taupm) \leq \sqrt{\tau_q}$ condition in above Proposition \ref{pro:joint} we require $\beta' \leq \sqrt{\beta}$ in (\ref{eqn:ber_claim}); this is satisfied in the Bernoulli case.
The proof of Proposition \ref{pro:joint} uses a technique called the \emph{Ahlswede-Winter method}~\cite{Ahlswede2002}, that results the factor of $k^2$ appearing in (\ref{eqn:joint}). 
The counterpart result for the marginal case is as follows.

\newcommand{\TroppThm}{C}
\begin{thmN}{\TroppThm,~\textit{c.f.}, Thm. 5.1,~\cite{Tropp2011}} 
Let the assumptions on matrix $\A$, index subset $\S$, and matrix $\pmb{X}$ be the same as Theorem \ref{pro:joint}. 
Let $\taupM \define \eigM(\E \pmb{X})$ and $\taupm \define \eigm(\E \pmb{X})$.
Let $\mathcal{D}(\cdot || \cdot)$ denote binary information divergence.
Then the following tail probability bounds hold
\ifthenelse{\boolean{dcol}}{ 
\bea
\Pr\{\sigM^2(\pmb{A}_\S) > \z\} &\leq& k  e^{-m\cdot \mathcal{D}(\frac{a}{k}||\taupM) }, \nn
\Pr\{\sigm^2(\pmb{A}_\S) < \z\} &\leq& k  e^{-m\cdot \mathcal{D}(\frac{a}{k}||\taupm) },\label{eqn:marg}
\eea
for respective limits $k\cdot \taupM < a < k$ and $0 < a < k\cdot \taupm$.}{
\bea
\Pr\{\sigM^2(\pmb{A}_\S) > \z\} &\leq& k  e^{-m\cdot \mathcal{D}(\frac{a}{k}||\taupM) },~~\mbox{ $k\cdot \taupM < a < k$} \nn
\Pr\{\sigm^2(\pmb{A}_\S) < \z\} &\leq& k  e^{-m\cdot \mathcal{D}(\frac{a}{k}||\taupm) }~~~~~~~\mbox{ $0 < a < k\cdot \taupm$}.\label{eqn:marg}
\eea}
\end{thmN}
The exponent in the bounds for the joint case (\ref{eqn:joint}), seem to be twice that of the marginal case (\ref{eqn:marg}). This would be true if the constant $c_3$ in (\ref{eqn:joint}) and (\ref{eqn:jointc3}) is small, and if $\tauAB/\taupM \approx \taupM$ and $\tauAB/\taupm \approx \taupm$ (the latter two conditions are true if $\beta \approx \beta'$ in (\ref{eqn:ber_claim}) above).
Recall that having (\ref{eqn:joint}) drop twice as fast as (\ref{eqn:marg}) is excellent from the standpoint of achieving a small Poisson approximation error $\eps_n(a)$.
Figure \ref{fig:RateCurve}$(a)$ seems to suggest that the exponent of (\ref{eqn:joint}) becomes double that of (\ref{eqn:marg}). 
Here we plot the exponents within the $\exp(\cdot)$ terms in both (\ref{eqn:joint}) and (\ref{eqn:marg}), where the exponent of (\ref{eqn:joint}) is ``halved'' for easier comparison (meaning that it is the exponent after factoring out $-2 m$, where for (\ref{eqn:marg}) we only factor out $-m$).
The $\sigM^2$ and $\sigm^2$ cases are respectively shown for different $a$ values in the ranges $a > 1$ and $a <  1$, 
according to the different given expressions for the ranges of $a$ (note $k \cdot \taupM = k \cdot \taupm = 1$).
As it becomes more apparent as $k$ increases from 4 to 20, the plotted (``halved'') exponents of (\ref{eqn:joint}) is close to that of (\ref{eqn:marg}).


\ifthenelse{\boolean{dcol}}{ 
\begin{figure}[!t]
	\centering
	  \epsfig{file=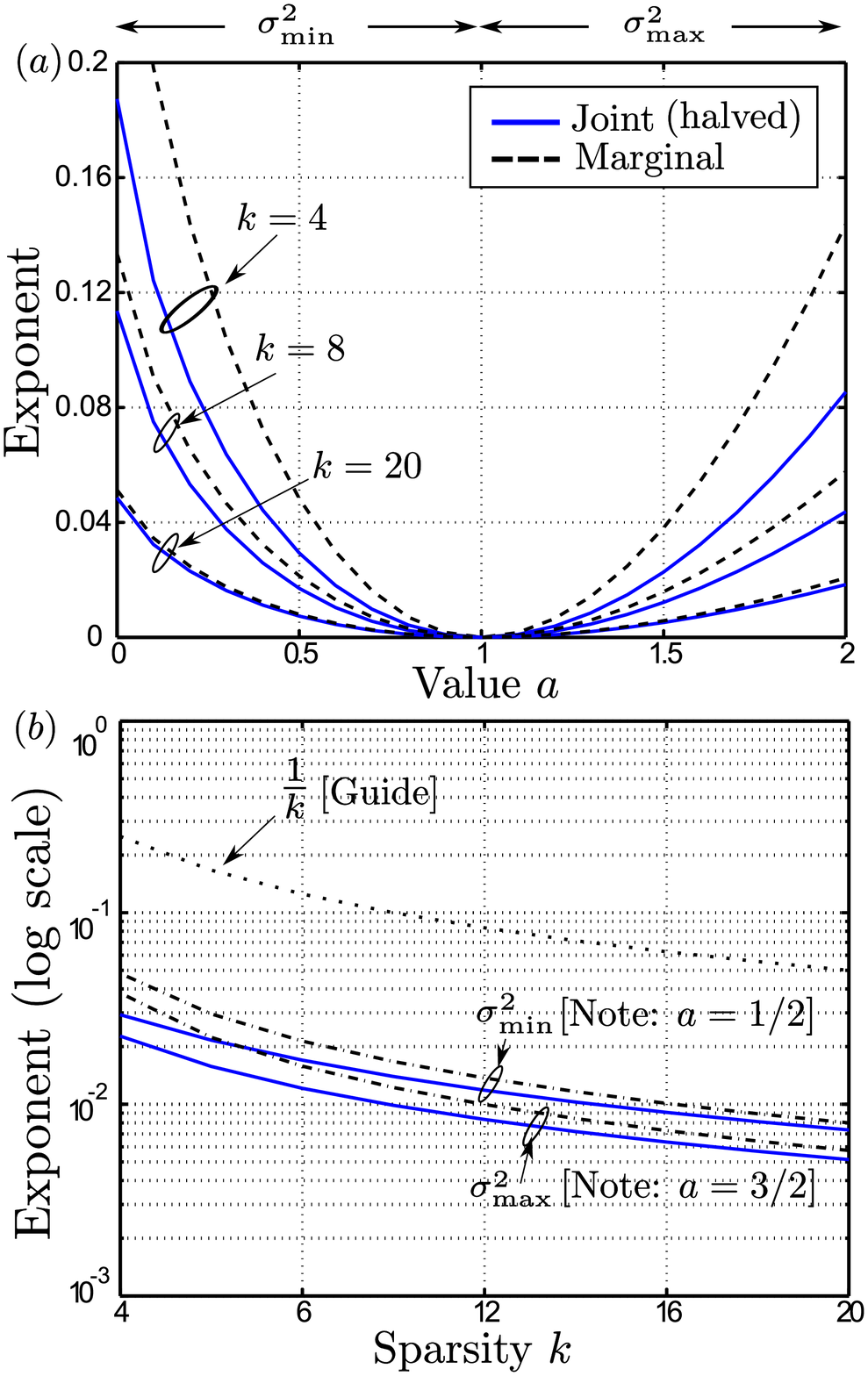,width=.7\linewidth}
		\caption{Bernoulli measure. In $(a)$, exponents of the joint and marginal tail probabilities (\ref{eqn:joint}) and (\ref{eqn:marg}) are plotted with respect to various $a$ values. In $(b)$, the same exponents are shown with respect to various $k$ (whereby $a$ is fixed).}
		\label{fig:RateCurve}
		\vspace*{-15pt}
\end{figure}
}{
\begin{figure}[!t]
	\centering
	  \epsfig{file=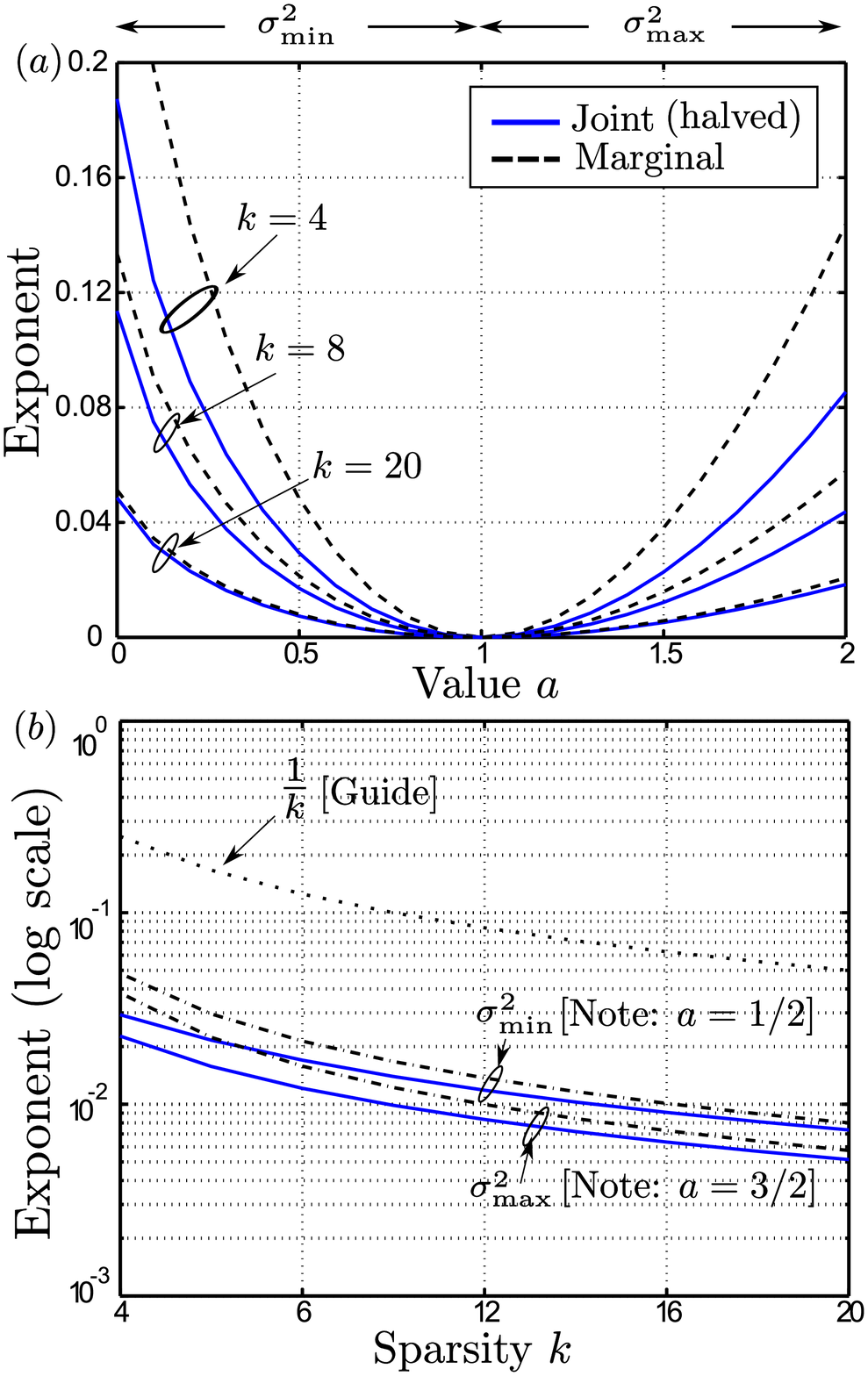},width=.8\linewidth}
		\caption{Bernoulli measure. In $(a)$, exponents of the joint and marginal tail probabilities (\ref{eqn:joint}) and (\ref{eqn:marg}) are plotted with respect to various $a$ values. In $(b)$, the same exponents are shown with respect to various $k$ (whereby $a$ is fixed).}
		\label{fig:RateCurve}
\end{figure}}

\newcommand{\const}{\beta}
The previous discussion is only aimed at developing intuition, and is not a proof of any sort.
We now evaluate the constant $c_3$ more carefully.
For the $\func = \sigM^2$ case (the $\func = -\sigm^2$ case follows similarly), we notice the following from the claim (\ref{eqn:ber_claim}): i) $c_1$ is proportional to $1/k$, where $c_1 = \tauAB/\taupM $, and ii) $c_2 $ is constant, where $c_2 = \taupM^2/\tauAB$. 
In particular for the Bernoulli case we will have $c_1 = 3/k$ and $c_2 = 1/3$. 
Under i) and ii) we claim that for large $k$, the constant $c_3$ is approximately $a/k$ (omitting a constant factor).
This implies that $c_3$ is gets smaller for fixed $a$ and increasing $k$, supporting the discussion in the previous paragraph.
To show this, we first argue that (omitting constant factors) $c_4$ is also approximately $a/k$.
From (\ref{eqn:jointc4}), by the Taylor approximation $\sqrt{1+\alpha} = 1 + 1/2 \alpha + o(\alpha)$ for $|\alpha| < 1$, we have $c_4 = \alpha/4 + o(\alpha)$ where here
\ifthenelse{\boolean{dcol}}{ 
\begin{align}
\alpha =& \frac{\const_4\cdot (1-\frac{a}{k})\frac{a}{k}}{(1-\pO)^2}
= \const_4 \left(\frac{a}{k}\right) \left(1-\left(\frac{a}{k}\right)\right) \left(1 + c_1 + o\left(\frac{1}{k}\right)\right)^2 \nn
= &\const_4 \cdot \frac{a}{k}+ o\left(\frac{1}{k}\right),
\label{eqn:const}
\end{align}}{
\bea
\alpha = \frac{\const_4\cdot (1-\frac{a}{k})\frac{a}{k}}{(1-\pO)^2}
= \const_4 \left(\frac{a}{k}\right) \left(1-\left(\frac{a}{k}\right)\right) \left(1 + c_1 + o\left(\frac{1}{k}\right)\right)^2 
= \const_4 \cdot \frac{a}{k}+ o\left(\frac{1}{k}\right),
\label{eqn:const}
\eea}
for $\const_4 = 4(c_2\Iv - 1)$, and ii) implies $\const_4$ to be some positive constant.
Thus $c_4 = (\const_4/4) \cdot (a/k) + o(1/k)$.
Moving on to (\ref{eqn:jointc3}), the above discussion shows that $(c_4 + a/k)/(a/k) = \const_4/4 + o(1)$, and
\ifthenelse{\boolean{dcol}}{ 
\bea
(1 + c_4)^2 &=& 1 + 2 \const_4 \cdot \frac{a}{k} + o\left(\frac{1}{k}\right),  \nn
\left(\frac{1-\frac{a}{k}}{1-\pO}\right)^2 &=& 1 + c_1 - 2\cdot \frac{a}{k} +o\left(\frac{1}{k}\right), \nonumber
\eea}{
\bea
(1 + c_4)^2 &=& 1 + 2 \const_4 \cdot \frac{a}{k} + o\left(\frac{1}{k}\right), ~~~\mbox{ and }~~~
\left(\frac{1-\frac{a}{k}}{1-\pO}\right)^2 = 1 + c_1 - 2\cdot \frac{a}{k} +o\left(\frac{1}{k}\right), \nonumber
\eea}
where the final identity follows similarly as in (\ref{eqn:const}).
Then $c_3 = \const_3 \cdot(a/k) + o(1/k)$, where $\const_3 = \const+ \log((\const_4/4)+o(1)) $ for constants $\const_3,\const$ .
Hence the claim that $c_3$ drops reciprocally in $k$ follows. 
%
%
%

\renewcommand{\fn}{\footnote{We ignored a $\log k$ term (due to $k^2$ in (\ref{eqn:joint})) on the RHS.}}
\newcommand{\constB}{\const_{\mathcal{D}}}
Unfortunately under our assumptions, the above argument that $c_3$ is small, is insufficient to show that for fixed $a$ the quantity $q_{k-1}(a)$ drops twice as fast as $p(a)$. 
This is because one can show (similarly as we did above) that $\mathcal{D}(a/k || c_1)$ also drops reciprocally in $k$, whenever $c_1$ is proportional to $1/k$. 
That is for fixed $a$, the exponents of both joint and marginal bounds (\ref{eqn:joint}) and (\ref{eqn:marg}) get smaller as $k$ increases, as illustrated in Figure \ref{fig:RateCurve}$(b)$.
This figure plots essentially the same exponents shown in Figure \ref{fig:RateCurve}$(a)$, but here $a$ is fixed to two values ($a=3/2$ for $\sigM^2$ and $a=1/2$ for $\sigm^2$), and the horizontal axis is now w.r.t. $k$.
We see that as the exponents shown drop at an approximate rate of $1/k$ with increasing $k$.
While the techniques behind Theorem \TroppThm~(and also Proposition \ref{pro:joint}) are simple, as pointed out in~\cite{Tropp2011} that, they do worse than Theorem \LedThm~if the columns $\pmb{A}_i$ are IID (which we assume in Theorem \PoiThm). 
Nevertheless we can show that the Poisson approximation error $\eps_n(a)$ in (\ref{eqn:errPA2}), drops if we allow $a$ to grow. 
As mentioned before, the main concern is the exponent of the second term in (\ref{eqn:errPA2}), where 
the constant term in front of $q_{k-1}(a)$ has exponent (at most) $(2k-1) \cdot [ 1 + \log (n/(k-1))]$.
Taking the exponent of $q_{k-1}(a)$ (w.r.t. to $-2m$) as $\mathcal{D}(a/k || c_1) + c_3$.
Approximate as before $\mathcal{D}(a/k || c_1) = \constB \cdot(a/k) + \mathcal{O}(1/k)$ for some constant $\constB$, also as before $c_3 = \const_3 \cdot(a/k) + o(1/k)$, hence we require\fn 
\bea
m \cdot  \left( \const \cdot \left(\frac{a}{k}\right) + \mathcal{O}\left(\frac{1}{k}\right) \right) 
> \left(k-\frac{1}{2}\right) \left[ 1 + \log \left( \frac{n}{k-1}\right) \right] \nonumber 
\eea
where $\const = \constB + \const_3$. 
%
Simply taking $k-1/2$ and $k-1$ as $k$ (assuming moderately large $k$), we essentially proved the following main result of this subsection.

\renewcommand{\fn}{\footnote{Like most results that require independence, there is a natural generalization to martingales, e.g. see~\cite{Tropp2011}.}}


\begin{thm} \label{cor:DDD}
Assume that the columns $\pmb{A}_i$ of $\pmb{A}$ are IID. 
Consider the error $\eps_n(a)$ in (\ref{eqn:errPA}) for the restricted isometries case, i.e., $\func = \sigM^2$ or $\func= -\sigm^2$.
Assume the terms $\tauAB, \taupM$ and $\taupm$ (as defined in Proposition \ref{pro:joint}) satisfy (\ref{eqn:ber_claim}).
Let $m =  t_1\cdot k (1 + \log((n/k)))$ for some constant $t_1$. 
Then the error $\eps_n(a)$ in (\ref{eqn:errPA}) will exponentially drop in $m$, if we set $a = (t_2/\const) \cdot k$, where constants $t_2,\const$ satisfying $t_2 < t_1^{-1}$ and $\const = \constB + \const_3$ with $\constB$ and $\const_3$ corresponding to respective approximations of the above terms $\mathcal{D}(a/k || c_1)$ and $c_3$, and where $k$ is sufficiently large.
\end{thm}


Explicit constants could be obtained by more careful book-keeping, but not done here for brevity considerations.
Note we require $k$ to be sufficiently large to allow to previous $\mathcal{O}(1/k)$ term to become small enough (\textit{i.e.}, to allow $t_2 + \mathcal{O}(1/k) \leq t_1^{-1}$). 
Incidentally, recall that previous Figure \ref{fig:RateCurve}$(a)$ showed joint (``halved'') and marginal exponents for as $k$ increases.

We conjecture that we should be able to improve Theorem \ref{cor:DDD} to hold for \emph{fixed} $a$ - we leave this to future research. 
\ifthenelse{\boolean{longver}}{
Figure \ref{fig:MaxEig} shows experimental results supporting theory derived in this subsection. 
In the spirit of non-asymptotics, and also for experimental complexity considerations, we choose moderate measurement size $m= 50$ and four different block lengths $n$ of $100, 200, 500$ and $1000$.
Corresponding to these values for $n$, Figures \ref{fig:MaxEig}$(a)-(d)$ plots the empirical ``tail probability`` for the maximal squared singular value $\sigM^2$ case (the other minimum case $\sigm^2$ is omitted for brevity - though similar observations are expected).
Note that this ``tail probability'' is actually obtained using Algorithm 1 (\emph{Greedy singular value pursuit)} from~\cite{Dossal2009}, as the actual quantity is too computational intensive to compute. 
We also plot the ``union bound'' $1 - \exp(-\Lam(\z))$ where the marginal quantity $p(a)$ in $\Lam(\z)$ is obtained via \emph{extrapolation}.
More specifically for $p(a) = \Pr\{\sigM^2(\pmb{A}_\S) > a\}$ , this probability was simulated for small $m$ values (for which this event can be observed, and then extrapolated to $m=50$  (see Supplementary Material \ref{app:extrap}). 
This was done (rather than using bound in Theorem \LedThm) in attempt to obtain very tight predictions.
Assuming both Algorithm 1 and the extrapolation work well, the results suggest good match (between empirical values and union bound prediction) for the various cases $k=4 \sim 8$ shown here. In particular the case $k=4$ is remarkably excellent. 
}{ 
This conjecture is inspired by recent work~\cite{Dossal2009}, whereby for fixed values of $a$ that satisfy recovery guarantees (similar to Theorem \RIPthm), see equation (17),~\cite{Dossal2009}, 
experimental validation of theoretical union bounds presented in~\cite{Blanchard} (similar to (\ref{eqn:RICprob})) has been performed.
The reader is referred to~\cite{Dossal2009} for these results, performed for an undersampling ratio $m/n = 1/4$, and a wide range of values $m = 250 \sim 2000$ and $n = 1000 \sim 8000$.
Also to support this conjecture for smaller $m=50$, in Supplementary Material \ref{app:extrap} we present experimental results. 
}

The rest of the subsection discusses the proof techniques.
Although Theorem \ref{pro:joint} looks significantly more complicated than Theorem \TroppThm, the proof techniques that follow the Ahlswede-Winter method are very similar. To facilitate the proof of Proposition \ref{pro:joint}, we first present the proof for Theorem \TroppThm. Our proof is slightly different (and simpler) than that in~\cite{Tropp2011}, due to the fact that we made a further simplifying assumption that the rows of $\pmb{A}$ are IID; in~\cite{Tropp2011} the independence\fn~assumption also holds but the identical assumption is not necessary.



\renewcommand{\fn}{\footnote{Should have a quick comment on the convergence of this sum.}}



\newcommand{\Smat}{\pmb{S}}
\newcommand{\X}{\pmb{X}}
\newcommand{\h}{h}

\newcommand{\Tmat}{\pmb{T}}
\renewcommand{\A}{\mat{C}}
\newcommand{\B}{\mat{D}}
\newcommand{\Y}{\pmb{Y}}
\newcommand{\tauqM}{\tau_{q}}
\newcommand{\tauqm}{\tau_{q}}

\ifthenelse{\boolean{dcol}}{ 
\begin{figure*}[bp]
	\normalsize
	\setcounter{tempequationcounter}{\value{equation}}
	\hrulefill
	\bea
		\setcounter{equation}{24}
		\!\!\!\E\{\Tr (\A e^{\h_1\X} )  \Tr(\B e^{\h_2\Y} ) \} 
		\!\!\!&\leq& \!\!\!
		\Tr(\A) \Tr(\B) \left\{ \tauqM (e^{\h_1}-1)(e^{\h_2}-1) 
		 +   \taupM (e^{\h_1}-1)  + \taupM (e^{\h_2}-1) + 1 \right\}, \nn
		 \!\!\!\E\{\Tr (\A e^{\h_1\X} )  \Tr(\B e^{\h_2\Y} ) \} 
		 \!\!\!&\leq& \!\!\!
		 \Tr(\A) \Tr(\B) \left\{ \tauqm (e^{-\h_1}-1)(e^{-\h_2}-1) +   \taupm (e^{-\h_1}-1)   +~\taupm (e^{-\h_2}-1) + 1 \right\} ,
		 \label{eqn:joint_iter}
	\eea
	\setcounter{equation}{\value{tempequationcounter}}
\end{figure*}}

Consider the $m\times k$ submatrix $\pmb{A}_\S$ of $\pmb{A}$. Express the product $\pmb{A}_\S^T \pmb{A}_\S$ as an average of $m$ random matrices $\pmb{X}_i$, i.e. express  $\pmb{A}_\S^T \pmb{A}_\S = \frac{1}{m}\sum_{i=1}^m \pmb{X}_i $ whereby $\pmb{X}_i$ is the $i$-th row outer product of $m \cdot \pmb{A}_\S$. Clearly if $|\S|=1$, then each $\pmb{X}_i$ becomes a scalar RV and  $\pmb{A}_\S^T \pmb{A}_\S$ becomes simply an average of scalar RVs. 
The Ahlswede-Winter method is essentially a concentration result for sums of random matrices. Let $\Smat_m$ denote the matrix sum that satisfies $\Smat_m = \sum_{i=1}^m\X_i$. Write $\pmb{A}_\S^T \pmb{A}_\S = \frac{1}{m} \Smat_m$ and express similarly as in (\ref{eqn:Umax2})
\bea
\Pr\{\sigM^2(\pmb{A}_\S) > \z\} &=& \Pr\{\eigM(\Smat_m) > m \cdot a\}, ~~~~\mbox{ and} \nn
\Pr\{\sigm^2(\pmb{A}_\S) < \z\} &=& \Pr\{\eigM(-\Smat_m) > - m \cdot a\}, \nonumber
\eea
so that it suffices to only look at the maximal eigenvalue function $\eigM$. That is we will only need to treat the quantities $\eigM(\Smat_m)$ and $\eigM(-\Smat_m)$. For a real, symmetric matrix $\mat{A}$, let $e^{\mat{A}}$ is denote the \emph{matrix-exponential} that satisfies $e^{\mat{A}}=\sum_{i=0}^\infty \frac{1}{i!}\mat{A}^i$. If $\varsigma$ is an eigenvalue of $\mat{A}$, then $e^{\varsigma}$ is an eigenvalue of $e^{\mat{A}}$. 
By convexity of the function $e^{\alpha}$, the inequality $e^{h\alpha} \leq 1 + (e^{h } -1) \cdot \alpha $ holds for all $h \in \Real$ and $0 \leq \alpha \leq 1$. Let $\mat{I}$ denote the identity matrix. For any real, symmetric matrix $\mat{A}$ whereby $\eigM(\mat{A})\leq 1$, the properties of the matrix exponential and the inequality $e^{h\alpha} \leq 1 + (e^{h } -1) \cdot \alpha $ imply that for any $h \in \Real$
\bea
		\mat{I} + (e^h -1) \mat{A} - e^{h \mat{A}} \mbox{ is positive semidefinite }. \label{eqn:Cher}
\eea
Because $\Smat_m = \sum_{i=1}^m\X_i$ and $\X_i$ are row outer sums, therefore $\Smat_m$ is positive semidefinite. 
For any real, symmetric matrix $\mat{A}$, the matrix exponential $e^{\mat{A}}$ is clearly positive semidefinite. Also for a positive semidefinite matrix $\mat{A}$, we have $\Tr(\mat{A})\geq \eigM(\mat{A})$. For any $h,t> 0$, we have 
\ifthenelse{\boolean{dcol}}{ 
\begin{align}
 &\Pr\{\eigM(\Smat_m) \geq t \} \nn
 &= \Pr\{ e^{\h \cdot \eigM(\Smat_m)}\geq e^{\h t} \} 
 = \Pr\{ \eigM(e^{\h \Smat_m}) \geq e^{\h t} \}\nn
 &\leq    \Pr\{\Tr( e^{\h\Smat_m}) \geq e^{\h t}\}
~~\leq~ e^{-\h t} \E \Tr(e^{\h \Smat_m}), \label{eqn:chain}
\end{align}}{
\bea
 \Pr\{\eigM(\Smat_m) \geq t \} = \Pr\{ e^{\h \cdot \eigM(\Smat_m)}\geq e^{\h t} \} 
 &=& \Pr\{ \eigM(e^{\h \Smat_m}) \geq e^{\h t} \}\nn
 &\leq&    \Pr\{\Tr( e^{\h\Smat_m}) \geq e^{\h t}\}
\leq e^{-\h t} \E \Tr(e^{\h \Smat_m}), \label{eqn:chain}
\eea}
where the first inequality follows because $e^{\h\Smat_m}$ is positive semidefinite, and the second inequality follows from Markov's inequality. Similarly, for any $h,t> 0$ we also have
\[
\Pr\{\eigM(-\Smat_m) \geq -t \} \leq e^{\h t} \E \Tr(e^{ -\h\Smat_m}).
\]
The proof of Theorem \TroppThm~relies on the following lemma, shown using the fact (\ref{eqn:Cher}).

\begin{lem} \label{lem:corAW}
Let $\X$ be a random, positive semidefinite matrix that satisfies $\eigM(\X)\leq 1$. Let $\Am$ be any positive semidefinite matrix of the same size as $\X$. Then for any $h\geq 0$ we have the following inequalities
\bea
\Tr(\Am e^{h\X}) &\leq&  \Tr(\Am) + (e^h-1)\Tr(\Am\X), \nn
\Tr(\Am e^{h\X}) &\leq&  \Tr(\Am) + (e^{-h}-1)\Tr(\Am\X). \label{eqn:lemcorAW0}
\eea
Taking expectation, we also have
\bea
 \!\!\!\!\!\!\E \{\Tr(\Am e^{\h \X})\} &\leq &\Tr(\Am) \left(e^{\h}\taupM + (1-\taupM) \right) \nn
 \!\!\!\!\!\!\E \{\Tr(\Am e^{-\h \X})\} &\leq& \Tr(\Am) \left(e^{-\h}\taupm + (1-\taupm) \right) \label{eqn:lemcorAW1}
\eea
where $\taupM = \eigM(\E\X) $ and $\taupm = \eigm(\E\X) $.
\end{lem}

The proof of Lemma \ref{lem:corAW} is relegated to Appendix \ref{app:proofRI}. To show Theorem \TroppThm~we also need the \emph{Golden-Thompson inequality}~\cite{Ahlswede2002}. The Golden-Thompson inequality states that for any two real and symmetric matrices $\mat{A}$ and $\mat{B}$, we have $\Tr (e^{\mat{A} + \mat{B}}) \leq \Tr(e^{\mat{A} } e^{\mat{B}})$. The proof of Theorem \TroppThm~is also furnished in Appendix \ref{app:proofRI}. 

Proposition \ref{pro:joint} for the joint case is similarly proved using Lemma \ref{lem:corAW}. Consider two size-$k$ subsets $\S$ and $\R$, that intersect in exactly $i$ positions, \textit{i.e.}, $|\S\cap \R|=i$. In addition to $\Smat_m$, similarly define another matrix $\Tmat_m$ that satisfies $\Tmat_m= \sum_{i=1}^m \Y_i$, where each $\Y_i$ is a size $k\times k$ matrix. Also similar to $\X_i$, let $\Y_i$ equal the $i$-th row outer product of the matrix $\frac{m}{k}\cdot \pmb{A}_\R$. 
Recall the joint Markov inequality, where for two RVs $A$ and $B$ and for any $t_1,t_2> 0$, we have $\Pr\{A > t_1, B > t_2 \} \leq \frac{\E A B}{t_1 t_2}$. Applying similar reasonings as in (\ref{eqn:chain}) we get for $h_1,h_2>0$
\ifthenelse{\boolean{dcol}}{ 
\begin{align}
	\Pr\{& \eigM(\Smat_m) > t, \eigM(\Tmat_m) > t\} \nn
	&\leq \Pr\{\Tr (e^{\h_1 \Smat_m}) > e^{\h_1 t}, \Tr(e^{\h_2 \Tmat_m}) > e^{\h_2 t}\} \nn
	&\leq e^{- t(\h_1+\h_2)} \cdot \E\{\Tr (e^{\h_1 \Smat_m})  \Tr(e^{\h_2 \Tmat_m}) \}, ~~~~~~~\mbox{ and } \nn
	\Pr\{& \eigM(-\Smat_m) > -t, \eigM(-\Tmat_m) > -t\}  \nn
	&\leq e^{ t(\h_1+\h_2)} \cdot \E\{\Tr (e^{-\h_1 \Smat_m})  \Tr(e^{-\h_2 \Tmat_m}) \}. \label{eqn:chain2}
\end{align}}{
\bea
	\Pr\{ \eigM(\Smat_m) > t, \eigM(\Tmat_m) > t\} 
	&\leq& \Pr\{\Tr (e^{\h_1 \Smat_m}) > e^{\h_1 t}, \Tr(e^{\h_2 \Tmat_m}) > e^{\h_2 t}\} \nn
	&\leq& e^{- t(\h_1+\h_2)} \cdot \E\{\Tr (e^{\h_1 \Smat_m})  \Tr(e^{\h_2 \Tmat_m}) \}, ~~~~~\mbox{ and } \nn
	\Pr\{ \eigM(-\Smat_m) > -t, \eigM(-\Tmat_m) > -t\} 
	&\leq& e^{ t(\h_1+\h_2)} \cdot \E\{\Tr (e^{-\h_1 \Smat_m})  \Tr(e^{-\h_2 \Tmat_m}) \}. \label{eqn:chain2}
\eea}
Let $\X,\Y$ denote a random, positive semidefinite matrices of the equal size. For any positive semidefinite matrices $\A,\B$ same size as $\X$, apply (\ref{eqn:lemcorAW0}) in Lemma \ref{lem:corAW} and use similar arguments that appear in its proof (see Appendix \ref{app:proofRI}) 
\ifthenelse{\boolean{dcol}}{ 
to show for $h_1,h_2 > 0$ that (\ref{eqn:joint_iter}) (see page bottom) holds,}{
to show for $h_1,h_2 > 0$ 
\bea
\!\!\!\E\{\Tr (\A e^{\h_1\X} )  \Tr(\B e^{\h_2\Y} ) \} 
\!\!\!&\leq& \!\!\!
\Tr(\A) \Tr(\B) \left\{ \tauqM (e^{\h_1}-1)(e^{\h_2}-1) 
 +   \taupM (e^{\h_1}-1)  + \taupM (e^{\h_2}-1) + 1 \right\}, \nn
 \!\!\!\E\{\Tr (\A e^{\h_1\X} )  \Tr(\B e^{\h_2\Y} ) \} 
 \!\!\!&\leq& \!\!\!
 \Tr(\A) \Tr(\B) \left\{ \tauqm (e^{-\h_1}-1)(e^{-\h_2}-1) +   \taupm (e^{-\h_1}-1) \right. \nn
 &&~~~~~~~~~~~~~~~~~~~~~~~~~~~~~~~~~~~~~~~~~~~\left.  +~\taupm (e^{-\h_2}-1) + 1 \right\} ,
 \label{eqn:joint_iter}
\eea}
where the constant $\tauqM$ satisfies (\ref{eqn:tauAB}), \textit{i.e.}, satisfies $\tauqM \cdot\Tr(\A)\Tr(\B) \geq \E\Tr(\A\X)\Tr(\B\Y) $.
We are now ready to prove Proposition \ref{pro:joint}, given in detail in the Appendix \ref{app:proofRI}.



To finish up this subsection, we address how to compute a constant $\tauAB$ that satisfies the hypothesis (\ref{eqn:tauAB}) required in Proposition \ref{pro:joint}. 


{
\begin{pro} \label{pro:tau_est}
Let $\X$ be an outer product of the row $[A_1,A_2,\cdots,A_k]$ of $k$ RVs $A_i$. Let $\Y$ be an outer product of the row $[B_1,B_2,\cdots,B_k]$ of $k$ RVs $B_i$. For some positive integer $c \leq k$, assume i) $A_i = B_i$ for $i\leq c$, ii) the RVs $A_i$ are IID, and iii) the RVs $B_i$ are IID. Let $X_{ij}$ and $Y_{ij}$ denote the matrix entries of $\X$ and $\Y$, respectively; note that $X_{ij}=A_i A_j$ and $Y_{ij}=B_i B_j$. Assume $\E A_1 = \E B_1 = 0$. Then
\ifthenelse{\boolean{dcol}}{\bea
\setcounter{equation}{25}
\E {X}_{ij} {Y}_{\ell\omega} &=&
\left\{
	\begin{array}{lcl}
		\E {X}_{ii} {Y}_{\ell\ell} =  \E A_{1}^4 & \mbox{ if~~i) holds}   \\
		\E {X}_{ii} {Y}_{\ell\ell} = (\E A_1^2)^2 & \mbox{ if~~ii) holds}   \\
		\E {X}_{ij} {Y}_{ij} = (\E A_1^2)^2 & \mbox{ if~~iii) holds} \\
		~~~~~~~0                      & \mbox{otherwise}
	\end{array}
\right. \label{eqn:proTauAB}
\eea
where above conditions i)-iii) are as follows
\bitm
\item[i)] $i = j=\ell=\omega$ and $1\leq i \leq c$.
\item[ii)] $i = j,~\ell=\omega$ and $i \neq \ell$  and $1 \leq i,\ell \leq k$.
\item[iii)] $i\neq j$ and $\{i,j\}= \{\ell,\omega\}$ and $1\leq i,j\leq c$.  
\eitm
}{ 
\bea
\E {X}_{ij} {Y}_{\ell\omega} &=&
\left\{
	\begin{array}{lcl}
		\E {X}_{ii} {Y}_{\ell\ell} =  \E A_{1}^4 & \mbox{ if } & i = j=\ell=\omega \mbox{ and } 1\leq i \leq c\\
		\E {X}_{ii} {Y}_{\ell\ell} = (\E A_1^2)^2 & \mbox{ if } & i = j,~\ell=\omega \mbox{ and } i \neq \ell  \mbox{ and } 1 \leq i,\ell \leq k\\
		\E {X}_{ij} {Y}_{ij} = (\E A_1^2)^2 & \mbox{ if } & i\neq j \mbox{ and } \{i,j\}= \{\ell,\omega\} \mbox{ and } 1\leq i,j\leq c   \\
		~~~~~~~0                      & \mbox{otherwise}
	\end{array}
\right. \label{eqn:proTauAB}
\eea}
Also for any positive semidefinite matrices $\A,\B$ of size $k\times k$, we have the following inequality
\ifthenelse{\boolean{dcol}}{
\bea
	\E \Tr(\A\X) \Tr(\B\Y) &\leq&   \left[ \max\left(\E A_{1}^4,\left(\E A_1^2\right)^2   \right) \right.  \nn
	 && ~~+ \left. 2\left(\E A_1^2\right)^2  \right] \cdot \Tr(\A)\Tr(\B).\nonumber
\eea}{
\[
	\E \Tr(\A\X) \Tr(\B\Y) \leq   \left( \max\left(\E A_{1}^4,\left(\E A_1^2\right)^2   \right)  
	 + 2\left(\E A_1^2\right)^2  \right) \cdot \Tr(\A)\Tr(\B).
\]}
\end{pro}
}

\begin{proof}
The RVs $A_1,A_2,\cdots, A_k , B_{c+1},B_{c+2},\cdots, B_k$ are IID, and $A_i = B_i$ for $i\leq c$.
Because $X_{ij}=A_i A_j$ and $Y_{ij}=B_i B_j$, then $X_{ij} Y_{\ell\omega} = A_i A_j B_\ell B_\omega$. Assume there exists at least one index (say $i$) that does not equal any of the other indices (say $j,\ell,\omega$), then $\E X_{ij} Y_{\ell\omega} = 0$ (in this case $\E A_i A_j B_\ell B_\omega = (\E A_i) (\E A_j B_\ell B_\omega) =0$ by our independence assumption, and the assumption $(\E A_i)=0$. That is under our assumptions, the only cases whereby $\E {X}_{ij} {Y}_{\ell\omega}\neq 0$ are outlined in (\ref{eqn:proTauAB}).

Let $c_{ij}$ and $d_{ij}$ denote the matrix entries of $\A$ and $\B$, respectively. We get that
\ifthenelse{\boolean{dcol}}{ 
\begin{align}
&\E \Tr(\A{\X}) \Tr(\B{\Y}) \nn
&= \sum_{i,j=1}^k  \sum_{\ell,\omega=1}^k c_{ij} \E ({X}_{ij} {Y}_{\ell\omega})   d_{\ell\omega}\nn
&= \sum_{i,\ell=1}^k  c_{ii} \E ({X}_{ii} {Y}_{\ell\ell})   d_{\ell\ell} 
  + 2 \mathop{\sum_{i=1}^{c}}_{i \neq j}  c_{ij}\E ({X}_{ij} {Y}_{ij})   d_{ij}. 
\label{eqn:proTauAB2}
\end{align}}{ 
\bea
\E \Tr(\A{\X}) \Tr(\B{\Y}) 
&=& \sum_{i,j=1}^k  \sum_{\ell,\omega=1}^k c_{ij} \E ({X}_{ij} {Y}_{\ell\omega})   d_{\ell\omega}\nn
&=& \sum_{i,\ell=1}^k  c_{ii} \E ({X}_{ii} {Y}_{\ell\ell})   d_{\ell\ell} 
  + 2 \mathop{\sum_{i=1}^{c}}_{i \neq j}  c_{ij}\E ({X}_{ij} {Y}_{ij})   d_{ij}. 
\label{eqn:proTauAB2}
\eea}
By assumed positive definiteness of $\A$ and $\B$, we have $c_{ii} \geq 0$ and $d_{ii} \geq 0$. Also $\E X_{ii} Y_{\ell\ell} \geq 0$ for all $i,\ell\geq 0$, see (\ref{eqn:proTauAB}). Hence the first term of (\ref{eqn:proTauAB2}) is upper bounded by $(\max_{i,\ell=1}^k \E X_{ii} Y_{\ell\ell} )\cdot \Tr(\A)\Tr(\B)$. By our assumptions for all $i\leq j$ where $i,j \leq c$, we have $\E X_{ij}Y_{ij} = (\E A_1^2)^2$. The second term of (\ref{eqn:proTauAB2}) is upper bounded by $2 (\E A_{1}^2)^2 \cdot\Tr(\A\B) $ (this upper estimate independent of constant $c$), where $\Tr(\A\B) \leq \Tr(\A)\Tr(\B)$ by positive semidefiniteness of $\A$ and $\B$. 
\end{proof}

We now verify the claim $\taup = \beta/k^2$ in (\ref{eqn:ber_claim}), when the columns $\pmb{A}_i$ are independent.
Take $\X=\sqrt{\frac{m}{k}} \pmb{A}_\S$ and $\Y=\sqrt{\frac{m}{k}} \pmb{A}_\R$ as in Proposition \ref{pro:joint}, with bounded IID entries $|X_{ij}|\leq 1/k$ and $|Y_{ij}| \leq 1/k$.
Use these $\X$ and $\Y$ in Proposition \ref{pro:tau_est} 
to conclude that we can choose $\tauAB$ as $\tauAB =  \max (\E X_{11}^2,(\E X_{11})^2) + 2\cdot (\E X_{11})^2 $, 
which must be of the form $\beta/k^2$ since $|X_{11}|\leq 1/k$.
Finally we comment with independent columns, the condition $\max(\taupM,\taupm) \leq \sqrt{\tauAB} $ in Proposition \ref{pro:joint}, is easily satisfied. Recall this condition is equivalently $\beta' \leq \sqrt{\beta}$ see (\ref{eqn:ber_claim}), and simply take $\tauAB = \max((\beta'/k)^2,\beta/k^2)$.


\subsection{Mutual coherence case} \label{ssect:MC}

Next to emphasize generality of the Poisson approximation Theorem \PoiThm, we demonstrate a different application.
In some early seminal work before the introduction of restricted isometry-type analyses, a different CS parameter was considered. Let $\Sens$ denote a matrix with $n$ columns $\col_i$, that satisfies the normalization $\Norm{\col_i}{2} =1$. The \textbf{mutual coherence} (or simply coherence) of such a matrix $\Sens$ is measured by the following quantity
\bea
	\mathop{\max_{1 \leq i,j \leq n}}_{i\neq j} |\col_i^T \col_j|. \label{eqn:corr}
\eea
By definition the mutual coherence is a number $a$ in $\Real$ between $0$ and $1$. 
CS recovery guarantees are obtainable from knowledge of (\ref{eqn:corr}), see \textit{e.g.},~\cite{Grib,Elad,Tropp2004,Tropp}, whereby the guarantees get stronger if the coherence gets smaller. 
We can relate the coherence to restricted isometry using the \emph{Gershorgin circle theorem}. Let $\func$ equal the function (\ref{eqn:g0}). 
As mentioned in~\cite{Tropp}, p. 2, for a matrix $\mat{A}$ with $k$ number of columns, all unit-norm, we have 
$\func(\mat{A})\leq (k-1) \cdot a $ where $a$ equals the mutual coherence (\ref{eqn:corr}) of $\mat{A}$, see~\cite{Tropp}, p. 2. However the coherence of an $m\times n$ matrix cannot be very small; it is at least $\sqrt{\frac{(n-m)}{m(n-1)}}$, see~\cite{Tropp}. 
Many techniques \textit{e.g.},~\cite{Calder,Cand2008,Tropp2008,Barga} involve the mutual coherence, thus also for the sake demonstrating the utility of U-statistic theory, we devote this small subsection to Poisson approximation of mutual coherence.  
While\cite{Lao} recently considered a more complicated analysis for a more complicated setting, the exposition here is original, simplified, and framed in the context of CS.
Here we only consider size-$2$ subsets $\S$.
Define the kernel $\func : \Real^{m\times 2} \rightarrow \Real$ as
\bea
	\func(\mat{A}) = \left|\frac{\mat{a}_1^T \mat{a}_2}{\Norm{\mat{a}_1}{2} \cdot \Norm{\mat{a}_2}{2}} \right| \label{eqn:gmC}
\eea
where $\mat{A}$ has two columns $\mat{a}_1$ and $\mat{a}_2$. 
Here (unlike the restricted isometry case) we make effort to normalize porperly.
For an $m\times n$ matrix $\Sens$, the statistic $\max_\S \func(\SensS)$ equals the mutual coherence (\ref{eqn:corr}) of $\Sens$.
Let $g$ be the indicator kernel satisfying $g(\mat{A},a) = \Ind{\func(\mat{A}) \geq a}$. 
Then the corresponding U-statistic $\U_n(a)$ with previously defined indicator kernel $g$, is related to the mutual coherence because $\{U_n(a) = 0\} = \{\max_\S \func(\SensS) \leq a\}$. The mutual coherence is also a ``worst-case'' statistic, similar to the restricted isometries, and so the concept of Poisson approximation applies similarly. To apply Theorem \PoiThm, we require estimates for $p(a)$ and $q_i(a)$, whereby in this case $|S|=2$ so we only have $q_1(a)$. 

\newcommand{\f}{f}
Both $p(a)$ and $q_1(a)$ are similarly estimated. Let $\pmb{A}$ be an $m\times n$ random matrix, and assume its columns $\pmb{A}_i$ to be IID. Denote the probability $\f(a,\mat{b})$ as
\bea
\f(a,\mat{b}) = \Pr\left\{ \left|\left( \frac{\pmb{A}_1}{\Norm{\pmb{A}_1}{2}}\right)^T\mat{b} \right| > a \right\} \label{eqn:f}
\eea
Let $\S = \{i_1,i_2\}$ and $\R = \{i_2,i_3\}$ whereby $\S\cap\R = \{i_2\}$, and by conditioning on $\pmb{A}_{i_2}$ we have
\ifthenelse{\boolean{dcol}}{ 
\begin{align}
p(a) &= \Pr\{\kernel(\pmb{A}_\S) > \z\} ~~~~~~~~~~~~~~= \E \f\left( a, \frac{\pmb{A}_{i_2}}{\Norm{\pmb{A}_{i_2}}{2}}\right), \nn
q_1(a) &= \Pr\{\kernel(\Sens_\S) > \z,\kernel(\Sens_{\mathcal{R}}) > \z \} = \E \f\left(a,  \frac{\pmb{A}_{i_2}}{\Norm{\pmb{A}_{i_2}}{2}}\right)^2. \label{eqn:mC}
\end{align}}{
\bea
p(a) &=& \Pr\{\kernel(\pmb{A}_\S) > \z\} ~~~~~~~~~~~~~~~= \E \f\left( a, \frac{\pmb{A}_{i_2}}{\Norm{\pmb{A}_{i_2}}{2}}\right), \nn
q_1(a) &=& Pr\{\kernel(\Sens_\S) > \z,\kernel(\Sens_{\mathcal{R}}) > \z \} = \E \f\left(a,  \frac{\pmb{A}_{i_2}}{\Norm{\pmb{A}_{i_2}}{2}}\right)^2. \label{eqn:mC}
\eea}
The following proposition provides an exponential bound for $\f(a,\mat{b})$ in (\ref{eqn:f}). Here, $\mat{b}$ is any vector in $\Real^m$ whereby $\Norm{\mat{b}}{2} = 1$. 
We prove the following result under for both Gaussian and Bernoulli matrices, due to slight complications introduced by the normalization in (\ref{eqn:gmC}).


\begin{pro}\label{pro:MC}
Let $\pmb{A}_1$ be a length-$m$ random vector with IID entries with zero mean, whereby each entry is either Gaussian with variance $1/m$, or Bernoulli $\{-1/\sqrt{m},\sqrt{m}\}$. The probability $\f(a,\mat{b})$ in (\ref{eqn:f}), for $\Norm{\mat{b}}{2} = 1$, is upper bounded as $\f(a,\mat{b}) \leq 2 \exp(- m\cdot a^2/2)$.
\end{pro}

We defer the proof for a moment. Use Proposition \ref{pro:MC} in (\ref{eqn:mC}), whereby substituting $\pmb{A}_{i_2}/\Norm{\pmb{A}_{i_2}}{2} =\mat{b}$, we get that $p(a)\leq 2 \exp(- m\cdot a^2/2)$ and $q_1(a)\leq 4 \exp(- m \cdot a^2)$. 
Thus the exponent of $q_1(a)$ is twice as large as that of $p(a)$, which suggests small Poisson approximation error $\eps_n(a)$. 
Here (unlike the restricted isometries case) we do not use (\ref{eqn:errPA2}), but instead use the other bound  (\ref{eqn:errPA3}) in the appendix. 
We can verify the following result by setting $k=2$, and further bounding $\Bin{n}{2}- \Bin{n-2}{2} < 2n - 3$ and $\Bin{2}{1}\Bin{n-2}{1}\Bin{n}{2} < n^3$.

\begin{thm} \label{cor:mC}
Let $\pmb{A}$ be an $m\times n$ random matrix, whereby the columns $\pmb{A}_i$ are IID and 
either Gaussian or Bernoulli distributed as described in previous Proposition \ref{pro:MC}.
Denote $\Lam(\z) = (n(n-1)/2) \cdot p(a)$.
Then, the Poisson approximation error $\eps_n(a)$ is upper bounded as
\ifthenelse{\boolean{dcol}}{ 
\bea
	\eps_n(a) &\leq& (1-e^{-\Lam(\z)})   \cdot \left[ (4n-6)\cdot \exp \left(- \frac{m a^2}{2} \right) \right] \nn
	&& + (4 n^3) \cdot \exp(- m a^2)  . \label{eqn:corMC}
\eea}{
\bea
	\eps_n(a) \leq (1-e^{-\Lam(\z)})   \cdot \left[ (4n-6)\cdot \exp \left(- \frac{m a^2}{2} \right) \right]
	+ (4 n^3) \cdot \exp(- m a^2)  . \label{eqn:corMC}
\eea}
\end{thm}

\renewcommand{\fn}{\footnote{In~\cite{Cand2008}, this estimate is given as $\sqrt{(2 \log n)/m}$ however we believe that a $\sqrt{2}$ factor has been omitted.}}

Theorem \ref{cor:mC} indicates that the mutual coherence is well predicted union bounds. 
By $p(a)\leq 2 \exp(- m\cdot a^2/2)$, we have $\Lam(\z) $ at most $ n^2 \cdot \exp (- m a^2/2)$, and $\Lam(z)$ 
drops exponentially in $m$ if $ a > \sqrt{(4 \log n)/m}$ - this is a standard estimate\fn, see~\cite{Cand2008}. 
As before we are mostly concernted about
the the second term in (\ref{eqn:corMC}), which requires a weaker condition on $a$ to drop exponentially. More specifically we only need $a > \sqrt{(\log 4 + 3 \log n)/m}$ (weaker than previous condition as long as $n > 4$). 
To conclude this subsection, we show the proof of Proposition \ref{pro:MC}.

{
\newcommand{\At}{{\pmb{X}}}
\newcommand{\C}{\mat{C}}
\renewcommand{\b}{\mat{b}}
\begin{proof}[Proof of Proposition \ref{pro:MC}]
For notational simplicity let $\At = \pmb{A}_{1}/\Norm{\pmb{A}_{1}}{2}$. We will show $\Pr\{\At^T \b > a\}\leq \exp(- m\cdot a^2/2)$, the other case $\Pr\{\At^T \b < -a\} \leq \exp(- m\cdot a^2/2)$ follows by symmetry of the distribution of $\At$. First consider the case where the entries of $\pmb{A}_1$ is Gaussian distributed, then $\At$ is uniformly distributed on the surface of an $m$-dimensional hypersphere. For any $m \times m$ orthogonal matrix $\C$, i.e. $\C^T\C = \mat{I}$, then $\At^T \C$ has the same distribution as $\At$. So choose any $\C\in \Real^{m\times m}$ such that $\C\b=[1,0\cdots,0]^T$ then $\Pr\{\At^T \b > a\}= \Pr\{X_1 > a \}$ since $\Norm{\mat{b}}{2}=1$. Since $\Norm{\At}{2}=1$, then $\Pr\{X_1 > a \}$ is proportional to the surface area of the spherical cap $\{x_1> a: \Norm{\mat{x}}{2}=1\}$. This probability is upper bounded by $(1-a^2)^{\frac{m}{2}} $, see~\cite{Racke}, p. XIII-3, which is in turn upper bounded by $\exp(- m\cdot a^2/2)$. 


In the case where the entries of $\pmb{A}_1$ Bernoulli distributed, then $\Norm{\pmb{A}_1}{2}=1$ and every entry $X_i$ of $\At$ is independent. For sums $S_m = \sum_{i=1}^m Y_i$ of independent RVs $Y_i$ with $|Y_i| \leq c_i$, see~\cite{Hoef} eqn. (2.6), we have $\Pr\{\frac{1}{m}S_m > a\} \leq \exp(m^2 a^2/(2 \Norm{\mat{c}}{2})^2)$ whereby $\mat{c}=[c_1,c_2\cdots,c_m]$. Setting $\At^T\b = \frac{1}{m} S_m$, we have $|Y_i| \leq \sqrt{m}\cdot b_i $, and setting $c_i = \sqrt{m}\cdot b_i$ we have $\Norm{\mat{c}}{2}^2 = m\cdot \Norm{\mat{b}}{2}^2 = m$, since $\Norm{\mat{b}}{2}=1$. Thus,  $\Pr\{\At^T \b > a\} \leq \exp(- m\cdot a^2/2)$.
\end{proof}
}

\begin{figure*}[!t]
	\centering
	  \epsfig{file={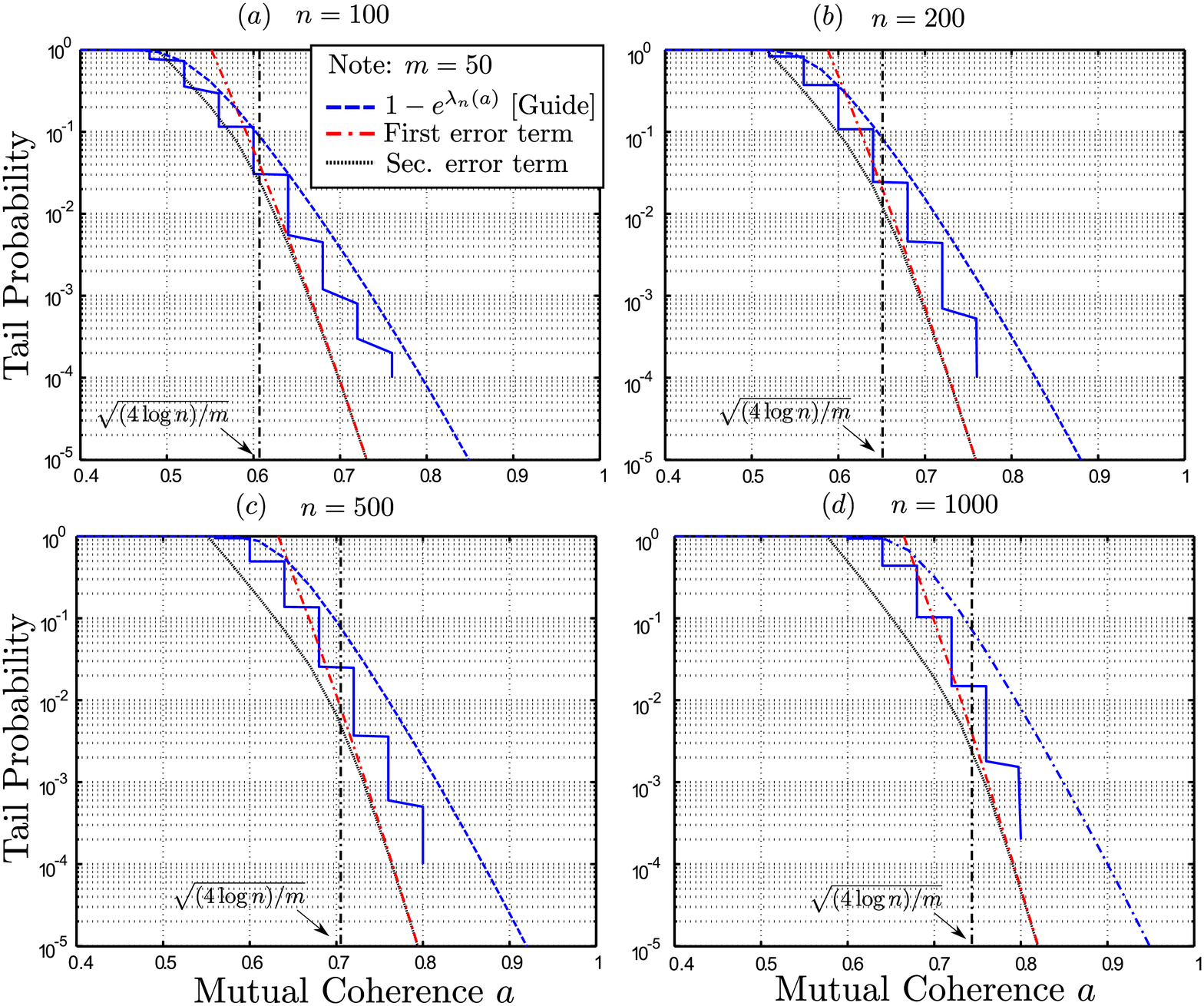},width=.7\linewidth}
		\caption{Bernoulli measure. For the mutual coherence case, comparing empirical tail probability with the ``union bound'' predicted by $1 - \exp(-\Lam(\z))$ (obtained using Gaussian approximation, see text). The error terms refer to (\ref{eqn:corMC}).}
		\label{fig:Corr}
		\vspace*{-15pt}
\end{figure*}

For the Bernoulli case, Figure \ref{fig:Corr} shows some empirical evidence that supports the theory derived in this subsection. 
Here we consider moderate measurement size $m=50$, and four different block lengths $n$ of $100, 200, 500$ and $1000$.
Corresponding to these values for $n$, Figures \ref{fig:Corr}$(a)-(d)$ plots the empirical tail probability of the mutual coherence, see (\ref{eqn:corr}). 
We also plot the function $1 - \exp(-\Lam(\z))$ where the marginal quantity $p(a)$ in $\Lam(\z)$ is taken to be $p(a) = 2\exp(-m a^2/2)/(a \sqrt{2 \pi}) $, and we expect this function to be close to the distribution of the mutual coherence.
This is due to the following two reasons. First from (\ref{eqn:gmC}) we see for the Bernoulli case $\func(\pmb{A}) = |\pmb{A}_1^T \pmb{A}_2|$, and $\pmb{A}_1^T \pmb{A}_2$ is a sum of $m$ IID Bernoulli $\{-1/m,1/m\}$ variables - which is 
\emph{approximately} Gaussian distributed with variance $1/m$.
Second, by Theorem \PoiThm~and Theorem \ref{cor:mC} we expect the mutual coherence to have good Poisson distribution. 
Indeed we observe for all cases of $n$ shown, that the empirical distribution is close to the plotted $1 - \exp(-\Lam(\z))$, \textit{i.e.}, the union bound is tight.
We also plot the error $\eps_n(a)$ in (\ref{eqn:corMC}); the first and second terms are plotted separately. 
We observe that for values of $a$ greater than the standard estimate $\sqrt{4\log n/m}$, the error values $\eps_n(a)$ become insignificantly small for all cases of $n$ shown. 

\section{Conclusion} \label{sect:conc}

This paper takes a first look at U-statistical theory for predicting the ``worst-case'' behavior of salient CS matrix parameters.
We showed how U-statistical theory is able to provide theoretical bounds on the tightness of union bounds analyses, 
whereby such results have never been investigated before in the CS context.
We investigated this premise for two important CS parameters: i) restricted isometries and ii) mutual coherence.
Our two main theorems determine that union bounds are tight, whereby 
for i) when $m = \mathcal{O}( k (1 + \log(n/k)))$ the restricted isometry constants need to grow linearly with sparsity $k$, and
for ii) the mutual coherence is of the standard estimate $\sqrt{(4\log n)/m}$.
That is under the specified conditions, the above two theorems justify the use of simple union bounds for ``worst-case'' analysis.

We discuss some directions for future work.
Firstly, it would be also desirable to improve the analyses in Subsection \ref{ssect:RI}, to allow the same conclusion for i) above but having the restricted isometry constants not depend on $k$. 
Secondly, it would be interesting to consider application of the techniques here to the null-space property, from which powerful recovery guarantees can be obtained. 
Thirdly, one might investigate the same tightness of union bound analyses for the case when the sampling matrix columns are dependent, whereby this requires appropriate extensions of Theorem \PoiThm.


\appendix

\subsection{Derivation of error estimate (\ref{eqn:errPA2})} \label{sup:error}

Here we derive (\ref{eqn:errPA2}) from (\ref{eqn:errPA}).
For some $a$ such that $p(a) >0$, note from (\ref{eqn:tail}) that $\Bin{n}{k} \cdot q_i(a)= \Lam(\z) p(a)^{-1} \jtail_i (\z) \geq (1-e^{-\Lam(\z)}) p(a)^{-1}\jtail_i (\z)$, the inequality follows because $1 - e^{-\alpha} \leq \alpha$ for all $\alpha \geq 0$.
Hence we can upper estimate the approximation error $\eps_n(a)$ given in (\ref{eqn:errPA}) as follows
\ifthenelse{\boolean{dcol}}{
\bea
	\eps_n(a) &\leq& (1-e^{-\Lam(\z)})\left\{\mtail(\z) \left[\Bin{n}{k}-\Bin{n-k}{k} \right]  \right\} \nn &&
	+ ~~\sum_{r=1}^{k-1}\Bin{k}{r}\Bin{n-k}{k-r} \Bin{n}{k} q_r(a).
	\label{eqn:errPA3}).
\eea}{
\bea
	\eps_n(a) \leq (1-e^{-\Lam(\z)})\left\{\mtail(\z) \left[\Bin{n}{k}-\Bin{n-k}{k} \right]  \right\} + \sum_{r=1}^{k-1}\Bin{k}{r}\Bin{n-k}{k-r} \Bin{n}{k} q_r(a).
	\label{eqn:errPA3}).
\eea}
We use a fact from\cite{Trust2011}, see Lemma 1, that for $i \leq k-1$ the inequality $q_{k-1}(\z) \geq q_i(\z)$ holds. By also using $\Bin{n-k}{k-1} \geq \Bin{n-k}{i}$ for all $i \leq k-1$, we claim 
\ifthenelse{\boolean{dcol}}{ 
\begin{align}
\sum_{r=1}^{k-1}&\Bin{k}{r}\Bin{n-k}{k-r} \Bin{n}{k} q_r(a) \nn
&\leq  \left[\sum_{r=1}^{k-1} \Bin{k}{r}\right] \Bin{n-k}{k-1} \Bin{n}{k} q_{k-1}(a) \nn
&\leq  2^k \cdot \Bin{2k-1}{k-1} \Bin{n}{2k-1}  q_{k-1}(a), \nn
&\leq  2^k \cdot \left( \frac{e(2k-1)}{k-1} \right)^{k-1} \cdot \left( \frac{e n}{2k-1} \right)^{2k-1} q_{k-1}(a). \nonumber 
\end{align}}{
\bea
\sum_{r=1}^{k-1}\Bin{k}{r}\Bin{n-k}{k-r} \Bin{n}{k} q_r(a) 
&\leq & \left[\sum_{r=1}^{k-1} \Bin{k}{r}\right] \Bin{n-k}{k-1} \Bin{n}{k} q_{k-1}(a) \nn
&\leq & 2^k \cdot \Bin{2k-1}{k-1} \Bin{n}{2k-1}  q_{k-1}(a), \nn
&\leq & 2^k \cdot \left( \frac{e(2k-1)}{k-1} \right)^{k-1} \cdot \left( \frac{e n}{2k-1} \right)^{2k-1} q_{k-1}(a). \nonumber 
\eea}
The second-last inequality follows from the identities $\sum_{r=1}^{k-1} \Bin{k}{r} =2^k-1$  and $\Bin{n-k}{i} \Bin{n}{k} = \Bin{k+i}{i} \Bin{n}{k+i}$.

\subsection{Technical proofs of claims appearing in Subsection \ref{ssect:RI}} \label{app:proofRI}

\begin{proof}[Proof of Lemma \ref{lem:corAW}]
Put $\pmb{B}=\mat{I} + (e^h -1) \X - e^{h \X}$. By the linearity of $\Tr(\cdot)$, we have $\Tr(\Am\pmb{B}) = \Tr(\Am) + (e^h-1)\Tr(\Am\X) - \Tr(\Am e^{h\X})$. Also since $\X$ is positive semidefinite, (\ref{eqn:Cher}) states that $\pmb{B}$ is positive semidefinite. For any two positive semidefinite matrices $\Am$ and $\mat{B}$, we have $\Tr(\Am\mat{B})\geq0$, therefore $\Tr(\Am e^{h\X}) \leq  \Tr(\Am) + (e^h-1)\Tr(\Am\X) $. Take expectations of both sides. Finally because $e^h -1 \geq 0$, use $\Tr(\Am \E\X) \leq \Tr(\Am) \cdot \eigM(\E\X) $ to prove the first inequality of (\ref{eqn:lemcorAW1}).

To show the second inequality, put $\pmb{B}=\mat{I} + (e^{-h} -1) \X - e^{-h \X}$. By (\ref{eqn:Cher}) this matrix $\pmb{B}$ is still positive semidefinite. The rest of the arguments follow similarly as the first case, however note that in this case $e^{-h} -1 \leq 0$ therefore we use $\Tr(\Am \E\X) \geq \Tr(\Am) \cdot \eigm(\E\X) $ to finish the proof for the second inequality of (\ref{eqn:lemcorAW1}).
\end{proof}

\ifthenelse{\boolean{dcol}}{
\begin{figure*}[bp]
	\normalsize
	\setcounter{tempequationcounter}{\value{equation}}
	\hrulefill
	\bea
		\setcounter{equation}{35}
		\pO e^\h &=& \frac{(1-\pO)(2(\tau+t)-1) + \sqrt{(1-\pO)^2 + 4(\pT\Iv-1)(\tau+t)(1-\tau-t)}}{2 (1-\tau-t)},
		 \label{eqn:dsads} \\
		 \pO e^\h &=&  \frac{(1-\pO)(\tau+t)}{1 - \tau - t} + \frac{-(1-\pO) + \sqrt{(1-\pO)^2 + 4(\pT\Iv-1)(\tau+t)(1-\tau-t)}}{2(1 - \tau - t)}\nn
  &=& \frac{(1-\pO)(\tau + c_4 +t)}{1 - \tau - t} \label{eqn:projoint1}
	\eea
	\setcounter{equation}{\value{tempequationcounter}}
\end{figure*}
}

\begin{proof}[Proof of Theorem \TroppThm]
In this proof, we set $\Smat_m= \sum_{i=1}^m \X_i$, where $\X_i$ is the $i$-th row outer sum of $\frac{m}{k} \pmb{A}_\S$, or simply $\pmb{A}_\S^T\pmb{A}_\S = \frac{d}{m} \Smat_m$. By the assumption $|A_{ij}|\leq 1$, then $\eigM(\X_i)\leq 1$. By (\ref{eqn:chain}) we have $\Pr\{\eigM(\Smat_m)> t \}\leq e^{-h t} \E \Tr(e^{h\Smat_m})$. First we want to show 
\bea
 \E \Tr(e^{h\Smat_m}) \leq k \left(e^h \taupM + 1- \taupM\right)^m. \label{eqn:promarg1}
\eea
Use the Golden-Thompson inequality to write $\E \Tr(e^{h\Smat_m}) \leq \E \Tr(e^{h\Smat_{m-1}} e^{h \X_m})$. For now use the notation shortcut $\tau =\taupM$. Because $\X_m$ is positive semidefinite and satisfies $\eigM(\X_m)\leq 1$, use Lemma \ref{lem:corAW} (for $e^{h\Smat_{m-1}}$ in place of $\Am$) to obtain $\E \Tr(e^{h\Smat_m}) \leq \E \Tr(e^{h\Smat_{m-1}})\left(e^{\h}\tau + (1-\tau) \right)  $. Repeat the argument $m-1$ more times for $\Smat_{m-1}, \Smat_{m-2}, \cdots $ where we finally get $\E \Tr(e^{h\Smat_m}) \leq \Tr(\mat{I}) \left(e^h \tau + 1- \tau\right)^m$ and $\Tr(\mat{I})=k$, showing (\ref{eqn:promarg1}). Putting previous facts together, for $t > 0$ we have the bound 
\[
\Pr\{\eigM(\Smat_m)> m\cdot(\tau + t) \}\leq k e^{-h m (\tau + t)}  \left(e^h \tau + 1- \tau\right)^m.
\]
Optimize the bound by setting $\tau e^h = (\tau+t)(1-\tau)/(1-\tau-t)$, see (4.7) in~\cite{Hoef}, where we enforce $\tau + t < 1$ to guarantee a positive solution for $h$. Using some manuipulations to express $\Pr\{\eigM(\Smat_m)> m\cdot(\tau + t) \}\leq k e^{\mathcal{D}(\tau + t||\tau)}$, where $\mathcal{D}(\cdot||\cdot)$ is the binary information divergence. Finally equate $\Pr\{\sigM^2(\pmb{A}_\S) \geq a\} = \Pr\{\eigM(\Smat_m) > \frac{ma}{k}\}$ and $\Pr\{\eigM(\Smat_m)> m\cdot(\tau + t) \}$ we set $t = a/k - \tau$, and we proved the first inequality, where the limits on $a$ follow from $\tau + t = a/k < 1$ and $t = a/k - \tau > 0$.

The second inequality is shown very similarly. For the rest of the proof, use the notation shortcut $\tau =\taupm$. 
Starting from the equation below (\ref{eqn:chain}), repeating similar arguments we can show 
\[
\Pr\{\eigM(-\Smat_m)> -m\cdot(\tau + t) \}\leq k e^{h m (\tau + t)}  \left(e^{-h} \tau + 1- \tau\right)^m,
\]
which is optimized by setting $\tau e^{-h} = (\tau+t)(1-\tau)/(1-\tau-t)$ to get $\Pr\{\eigM(-\Smat_m)> -m\cdot(\tau + t) \}\leq k e^{\mathcal{D}(\tau + t||\tau)}$, where we enforce $\tau + t> 0$ to guarantee a positive solution for $h$. Equating $\Pr\{\sigm^2(\pmb{A}_\S) < a\} = \Pr\{\eigM(-\Smat_m) > -\frac{ma}{k}\}$ and $\Pr\{\eigM(\Smat_m)> -m\cdot(\tau + t) \}$ we set $t = a/k - \tau$ to prove the second inequality, where the limits on $a$ follow from $\tau + t = a/k > 0$ and $t = a/k - \tau < 0$.
\end{proof}

\renewcommand{\fn}{\footnote{Alternatively, it might be easier to verify that when $t=0$, both of the two equations displayed above are satisfied when we set $h_1 = h_2 = 0$.}}
\newcommand{\p}{\taupM}
\newcommand{\pp}{\tauqM}
\begin{proof}[Proof of Proposition \ref{pro:joint}]
We continue where we left off from (\ref{eqn:chain2}). Apply the Golden-Thomson inequality on the term $\E\{\Tr (e^{\h_1 \Smat_m})  \Tr(e^{\h_2 \Tmat_m})\}$ to get 
\ifthenelse{\boolean{dcol}}{ 
\begin{align}
\E\{\Tr &(e^{\h_1 \Smat_m})  \Tr(e^{\h_2 \Tmat_m}) \} \nn
&\leq \E\{\Tr (e^{\h_1 \Smat_{m-1}}e^{\h_1\X_m} )  \Tr(e^{\h_2 \Tmat_{m-1}}e^{\h_2\Y_m} ) \}. \nonumber
\end{align}}{
\bea
\E\{\Tr (e^{\h_1 \Smat_m})  \Tr(e^{\h_2 \Tmat_m}) \} 
&\leq& \E\{\Tr (e^{\h_1 \Smat_{m-1}}e^{\h_1\X_m} )  \Tr(e^{\h_2 \Tmat_{m-1}}e^{\h_2\Y_m} ) \}. \nonumber
\eea}
Apply the first inequality in (\ref{eqn:joint_iter}), followed by the Golden-Thomson inequality and then (\ref{eqn:joint_iter}) again and so on, to show (using some further algebraic manipulations) that for $t > 0$
\ifthenelse{\boolean{dcol}}{ 
\begin{align}
\Pr&\left\{\eigM(\Smat_m)\geq m \left(\tau + t\right), \eigM(\Tmat_m) \geq m \left(\tau + t\right) \right\} \nn
&\leq k^2 e^{- m(\tau+t)(\h_1+\h_2)} \left\{ \pT(\pO e^{\h_1 } + 1 - \pO)(\pO e^{\h_2 } + 1 - \pO) \right.\nn
& ~~~~~~~~~~~~~~~~~~~~~~~~~~~\left.+ 1-\pT \right\}^m \label{eqn:projoint0}
\end{align}}{
\begin{align}
\Pr&\left\{\eigM(\Smat_m)\geq m \left(\tau + t\right), \eigM(\Tmat_m) \geq m \left(\tau + t\right) \right\} \nn
&\leq k^2 e^{- m(\tau+t)(\h_1+\h_2)} \left( \pT(\pO e^{\h_1 } + 1 - \pO)(\pO e^{\h_2 } + 1 - \pO) + 1-\pT \right)^m \label{eqn:projoint0}
\end{align}}
where the constants $\pO =  \pp/\p$ and $\pT = \p^2/\pp$, and we used the shorthand $\tau = \taupM$. Differentiating the exponent of (\ref{eqn:projoint0}) with respect to both $h_1$ and $h_2$ we get respectively
\ifthenelse{\boolean{dcol}}{ 
\begin{align}
	(\pO e^{\h_2 } + 1 - \pO) & \left(\pO e^{\h_1} \left(1-\tau - t\right) -  (1-\pO)\left(\tau+t\right)\right) \nn
	&= (\pT\Iv-1)(\tau+t),\nn
	(\pO e^{\h_1 } + 1 - \pO) & \left(\pO e^{\h_2} \left(1-\tau - t\right) -  (1-\pO)\left(\tau+t\right)\right) \nn
	&= (\pT\Iv-1)(\tau+t).
	 \nonumber
\end{align}}{
\bea
	(\pO e^{\h_2 } + 1 - \pO) \left(\pO e^{\h_1} \left(1-\tau - t\right) - (1-\pO)\left(\tau+t\right)\right)&=& (\pT\Iv-1)(\tau+t),\nn
	(\pO e^{\h_1 } + 1 - \pO) \left(\pO e^{\h_2} \left(1-\tau - t\right) - (1-\pO)\left(\tau+t\right)\right)&=& (\pT\Iv-1)(\tau+t).
	 \nonumber
\eea} 
To solve the previous two equations it suffices to have $h = h_1 = h_2$. Then by substituting $a=c_1 e^h$ we solve the quadratic equation $f(a) = a^2 + b a + c$ where $b = (1-\pO)(1-2(\tau + t))/(1-\tau -t)$ and $c = -(\tau + t)[(1-\pO) + (\pT\Iv - 1)]/(1-\tau-t)$. 
\ifthenelse{\boolean{dcol}}{ 
Under the assumption $\pT \leq 1$, the solution (\ref{eqn:dsads}) (see page bottom) for $a = \pO e^h $ will exist for some positive $h > 0$, if we constraint $\tau + t < 1$. To see this, check that for $0 < t < 1 - \tau$ the RHS above increases monotonically as $t$ increases, and verify\fn~that if we set $t=0$ the RHS of (\ref{eqn:dsads}) equals $c_1$ (in which case $h=0$).
We then write (\ref{eqn:projoint1}) (see page bottom) where $c_4 = c_4(a,k,\pO,\pT)$ is given as in (\ref{eqn:jointc4}) after equating $a/k = \tau + t$.}{ 
Under the assumption $\pT \leq 1$, the following solution for $a = \pO e^h $ will exist for some positive $h > 0$
\[
  \pO e^\h = \frac{(1-\pO)(2(\tau+t)-1) + \sqrt{(1-\pO)^2 + 4(\pT\Iv-1)(\tau+t)(1-\tau-t)}}{2 (1-\tau-t)},
\]
if we constraint $\tau + t < 1$. To see this, check that for $0 < t < 1 - \tau$ the RHS above increases monotonically as $t$ increases, and verify\fn~that if we set $t=0$ the RHS above equals $c_1$ (in which case $h=0$).
We write 
\bea
  \pO e^\h &=&  \frac{(1-\pO)(\tau+t)}{1 - \tau - t} + \frac{-(1-\pO) + \sqrt{(1-\pO)^2 + 4(\pT\Iv-1)(\tau+t)(1-\tau-t)}}{2(1 - \tau - t)}\nn
  &=& \frac{(1-\pO)(\tau + c_4 +t)}{1 - \tau - t} \label{eqn:projoint1}
\eea
where $c_4 = c_4(a,k,\pO,\pT)$ is given as in (\ref{eqn:jointc4}) after equating $a/k = \tau + t$.} 
Substituting (\ref{eqn:projoint1}) back into (\ref{eqn:projoint0}) and by further algebraic manipulations we get
\ifthenelse{\boolean{dcol}}{
\begin{align}
\Pr&\left\{\eigM(\Smat_m) \geq m \left(\tau + t\right), \eigM(\Tmat_m) \geq m \left(\tau + t\right) \right\} \nn
&~~~~~\leq k^2  e^{-m\cdot 2(\mathcal{D}(\tau+t||\pO) + c_3(t,k,\pO,\pT))} \nonumber
\end{align}}{
\[
\Pr\left\{\eigM(\Smat_m) \geq m \left(\tau + t\right), \eigM(\Tmat_m) \geq m \left(\tau + t\right) \right\} 
\leq k^2  e^{-m\cdot 2(\mathcal{D}(\tau+t||\pO) + c_3(t,k,\pO,\pT))}
\]}
where $c_3 = c_3(a,k,\pO,\pT)$ is given in (\ref{eqn:jointc3}) after equating $a/k = \tau + t$, and we get the desired result.
The limits on $a$ follow as before: $\tau + t = a/k < 1$ and $t = a/k - \tau > 0$.

For the other case we have $\pO =  \pp/\taupm$ and $\pT = \taupm^2/\pp \leq 1$, and we similarly show
\ifthenelse{\boolean{dcol}}{
\begin{align}
\Pr&\left\{\eigM(-\Smat_m)\geq -m \left(\tau + t\right), \eigM(-\Tmat_m) \geq -m \left(\tau + t\right) \right\} \nn
\leq& k^2 e^{m(\tau+t)(\h_1+\h_2)} \left\{ \pT(\pO e^{-\h_1 } + 1 - \pO)(\pO e^{-\h_2 } + 1 - \pO) \right. \nn
&~~~~~~~~~~~~~~~~~~~~~~+ \left. 1-\pT \right\}^m \nonumber
\end{align}}{
\begin{align}
\Pr&\left\{\eigM(-\Smat_m)\geq -m \left(\tau + t\right), \eigM(-\Tmat_m) \geq -m \left(\tau + t\right) \right\} \nn
&\leq k^2 e^{m(\tau+t)(\h_1+\h_2)} \left( \pT(\pO e^{-\h_1 } + 1 - \pO)(\pO e^{-\h_2 } + 1 - \pO) + 1-\pT \right)^m \nonumber
\end{align}}
where we use the shorthand $\tau = \taupm$. Again as in the second part of the proof of Theorem \TroppThm, we now constrain $\tau + t > 0$ and proceed similarly as before. Now treating $e^{-h}$ instead of $e^h$, the expression for $c_1 e^{-h}$ simply equals the RHS of the equation above (\ref{eqn:projoint1}), whereby the said RHS decreases monotonically as $t$ decreases in the range $-\tau < t < 0$. 
Hence we conclude as before that $c_1 e^{-h}$ equals the RHS of (\ref{eqn:projoint1}) for some positive $h$, and the essentially same expressions follow.
The limits on $a$ follow from $\tau + t = a/k > 0$ and $t = a/k - \tau < 0$.
\end{proof}



%



\end{document}